\documentclass[letterpaper]{article} 
\usepackage{aaai2026}  
\usepackage{times}  
\usepackage{helvet}  
\usepackage{courier}  
\usepackage[hyphens]{url}  
\usepackage{graphicx} 
\usepackage{natbib}  
\usepackage{caption} 
\usepackage{amsmath}
\usepackage{pifont}
\usepackage{amsthm}
\usepackage{dsfont}
\usepackage{bm}
\usepackage{graphicx}
\usepackage{amssymb}
\usepackage{array}

\newtheorem{theorem}{Theorem}

\newtheorem{lemma}{Lemma}
\newtheorem{assumption}{Assumption}
\newtheorem{definition}{Definition}
\newtheorem{corollary}{Corollary}
\newtheorem{observation}{Observation}

\newtheorem{example}{Example}

\usepackage{algorithm}
\usepackage{algorithmic}

\newcommand{\ind}{\perp\!\!\!\!\perp}
\newcommand{\reals}{\mathbb R}
\newcommand{\dd}{\mathrm{d}}
\copyrighttext{Code: \texttt{https://github.com/sidsrinivasan/womac}}
\usepackage{newfloat}
\usepackage{listings}
\DeclareCaptionStyle{ruled}{labelfont=normalfont,labelsep=colon,strut=off} 
\floatstyle{ruled}
\newfloat{listing}{tb}{lst}{}
\floatname{listing}{Listing}
\title{WOMAC: A Mechanism For Prediction Competitions}
\author{
  Siddarth Srinivasan\textsuperscript{\rm 1},
  Tao Lin\textsuperscript{\rm 1},
  Connacher Murphy\textsuperscript{\rm 2}\\
  Anish Thilagar\textsuperscript{\rm 3},
  Yiling Chen\textsuperscript{\rm 1},
  Ezra Karger\textsuperscript{\rm 2, 4}
}
\affiliations{
  \textsuperscript{\rm 1}Harvard University,
  \textsuperscript{\rm 2}Forecasting Research Institute,
  \textsuperscript{\rm 3}CU Boulder,
  \textsuperscript{\rm 4}Federal Reserve Bank of Chicago

}

\usepackage{bibentry}

\begin{document}

\maketitle

\vspace*{-26pt}
\begin{abstract}
  Competitions are widely used to identify top performers in judgmental forecasting and machine learning, and the standard competition design ranks competitors based on their cumulative scores against a set of realized outcomes or held-out labels. However, this standard design is neither incentive-compatible nor very statistically efficient. The main culprit is noise in outcomes/labels that experts are scored against; it allows weaker competitors to often win by chance, and the winner-take-all nature incentivizes misreporting that improves win probability even if it decreases expected score. Attempts to achieve incentive-compatibility rely on randomized mechanisms that \emph{add even more noise} in winner selection, but come at the cost of determinism and practical adoption. To tackle these issues, we introduce a novel \emph{deterministic} mechanism: WOMAC (Wisdom of the Most Accurate Crowd). Instead of scoring experts against noisy outcomes, as is standard, WOMAC scores experts against the best ex-post aggregate of peer experts' predictions given the noisy outcomes. WOMAC is also more efficient than the standard competition design in typical settings. While the increased complexity of WOMAC makes it challenging to analyze incentives directly, we provide a clear theoretical foundation to justify the mechanism. We also provide an efficient vectorized implementation and demonstrate empirically on real-world forecasting datasets that WOMAC is a more reliable predictor of experts' out-of-sample performance relative to the standard mechanism. WOMAC is useful in any competition where there is substantial noise in the outcomes/labels.
\end{abstract}


\section{Introduction}
\label{sec:intro}

Competitions are widely used to identify top performers across domains. Judgmental forecasting competitions elicit predictions on future events and pick one or a small handful of winners. In machine learning, public leaderboards rank models by performance on held-out validation sets, shaping perceptions of algorithmic progress and conferring both prestige and financial rewards. In both settings, a \emph{crowd of experts} (human or algorithmic) competes to achieve the highest score on a shared set of \emph{prediction tasks}. Success on platforms like Kaggle, Metaculus, or in high-profile contests such as the Netflix Prize yields substantial monetary and reputational benefits. Thus, good design is essential to accurately identify and reward top performers.

What makes for a good competition design? First, it should be \emph{incentive compatible (IC)}: each competitor must want to submit their true predictions without distortion. Second, it should be \emph{efficient}: the competition should identify top performers with high probability based on relatively few tasks. The former makes the competition useful to the designer by providing access to high-quality, accurate predictions or models. The latter is essential for identifying top performers, and also incentivizes stronger competitors to participate by assuring them of a fair chance of winning. A third criterion, important to practitioners but often overlooked, is that it is \emph{deterministic}: given a set of expert predictions, the competition always picks the same winner.

The standard mechanism used by competitions, such as the Good Judgment Project in the IARPA-funded ACE tournament \citep{mellers2014psychological} or Kaggle, scores each competitor's predictions against realized outcomes/held-out labels, and chooses the competitor with the best cumulative score as the winner. While this is a deterministic mechanism, it is neither incentive-compatible nor particularly efficient when there is noise in the outcomes/labels \citep{aldous2021prediction, frongillo2021efficient}. The challenge is that noisy outcomes, combined with the winner-take-all nature of competitions, induce experts to care about their \emph{relative} score and win probability instead of their \emph{absolute} score. Weaker experts can easily `luck into' winning, and rational experts may perversely be incentivized to skew their predictions to mimic weaker experts. Surprisingly, all proposed alternatives \citep{witkowski2023incentive, frongillo2021efficient} thus far give up determinism for incentive-compatibility, and thus require \emph{adding even more noise} to the selection of the winner.

\begin{algorithm}
  \caption{WOMAC}
  \textbf{Input:} $m \times n$ matrix ${\bf W}$ of expert predictions, $m \times 1$ vector $\vec{y}$ of outcomes
  \begin{algorithmic}[1]\label{alg:algorithm}
    \FOR{each task $i = 1$ to $m$}
    \FOR{each expert $j = 1$ to $n$}
    \STATE Learn $\beta_{ij} = \text{argmin}_{\beta'} \|\vec{y}_{-i} -  f_{\beta'}({\bf W}_{-ij}) \|^2$
    \STATE Compute ${t}_{ij} \gets f_{\beta_{ij}}(\vec{w}_{i,-j})$, a per-expert, per-task reference solution.
    \ENDFOR
    \ENDFOR
    \STATE Compute $S_j = \sum_i \|w_{ij} - t_{ij}\|^2$, scores for each expert
    \STATE \textbf{Return:} Winner $j^* = \displaystyle \mathop{\mathrm{arg\,min}}_j S_j$
  \end{algorithmic}
\end{algorithm}

Motivated by these challenges, this paper proposes a novel \emph{deterministic} mechanism for prediction competitions called \emph{WOMAC} (Wisdom of the Most Accurate Crowd). Instead of scoring against (noisy) realized outcomes, WOMAC uses a meta-learner---a predictive denoising model that fits experts' predictions to observed outcomes---to learn the best ex-post aggregation of experts' predictions on each task, and scores experts based on accuracy with respect to this aggregate.\footnote{For example, suppose the ground truth probability of an event is $\theta =0.05$, but the event happens so $y=1$. Under the standard mechanism, a strong expert who reasonably reported $w=0.06$ is unlucky when evaluated against the outcome. Under WOMAC, they are evaluated more fairly against an estimate like $\theta^s=0.04$.} In practice, this can mean taking a simple average of the top-$k$\% of experts by MSE against the realized outcomes, and scoring all experts against this average.  Importantly, WOMAC makes use of the information in the distribution of experts' reported predictions to arrive at a better reference to score experts against, instead of discarding this as in the standard mechanism. This also means it is more efficient than the benchmark standard mechanism in practical settings. While the increased complexity of WOMAC makes it challenging to analyze incentives directly, we nevertheless can provide solid theoretical foundations to justify the mechanism. Our proposed mechanism is useful in identifying top participants in both judgmental forecasting and ML competitions where the outcomes/labels are \emph{substantially noisy}.\footnote{This makes WOMAC appealing by default in binary outcome competitions where the 0/1 outcome is a noisy draw generated by some underlying Bernoulli parameter.}

This paper makes five contributions. First, we tackle a basic question: are deterministic prediction competitions \emph{inherently} not incentive-compatible, or is it due to scoring experts against noisy outcomes? We prove that deterministic competitions \emph{can be} Bayes-Nash incentive-compatible if experts are scored directly against a noiseless `ground truth' for a broad range of realistic distributions, establishing outcome/label noise as the challenge to incentives in deterministic competitions. Second, we show an intuitive result that scoring experts against less noisy estimates of the ground truth yields a higher probability of identifying the top expert, relative to scoring against noisier estimates. Third, we use these insights--that motivate scoring experts against less noisy estimates of the ground truth--to propose the WOMAC (wisdom of the most accurate crowd) mechanism; (4) we provide an efficient, vectorized implementation of our mechanism in PyTorch; and (5) we present experimental results on judgmental forecasting datasets showing that WOMAC is a better predictor of experts' scores and ranks on held-out questions than the standard competition.

This paper is structured as follows: we first present the background and related work; subsequently we describe the model and mechanism; next, we provide our theoretical results; then, we use the theoretical insights to develop the WOMAC mechanism; finally, we present experimental results using WOMAC. Proofs are deferred to the appendix.

\section{Background and Related Work}\label{sec:back}

\paragraph{Standard Mechanism}
As discussed in the introduction, an ideal competition should satisfy a \textbf{game theoretic} property (incentive-compatibility) and a \textbf{statistical} property (efficiency). In the standard mechanism for prediction competitions, experts' predictions are scored against outcomes using a strictly proper scoring rule (usually mean squared error), and the mechanism picks the expert with the lowest error as the winner. While competition mechanism efficiency is benchmarked relative to this standard mechanism, we can observe that this is limited by the noise in observations. Additionally, the mechanism is certainly not incentive-compatible (IC). Evaluating experts using MSE is IC when experts want to maximize their total proper score, but not when they want to maximize their win probability as in winner-take-all tournaments \citep{lichtendahl2007probability}. In the standard mechanism, experts can actually \emph{increase their win probability by misreporting predictions that reduce their expected score}. Even if experts reported truthfully \citet{aldous2021prediction} demonstrates the thorny statistical issues that still arise from noisy signals. Specifically, it's possible that by pure chance, some experts will have their noisy predictions align with the outcomes, thus lucking into a better score than the strongest experts and winning the competition (see Figure \ref{fig:gauss} for an example).  \citet{monroe2025hedging} do show that the standard mechanism can sometimes still meet both criteria, but only when there are a small number of skilled experts.

\begin{figure}[htbp]
  \centering
  \includegraphics[width=0.45\textwidth]{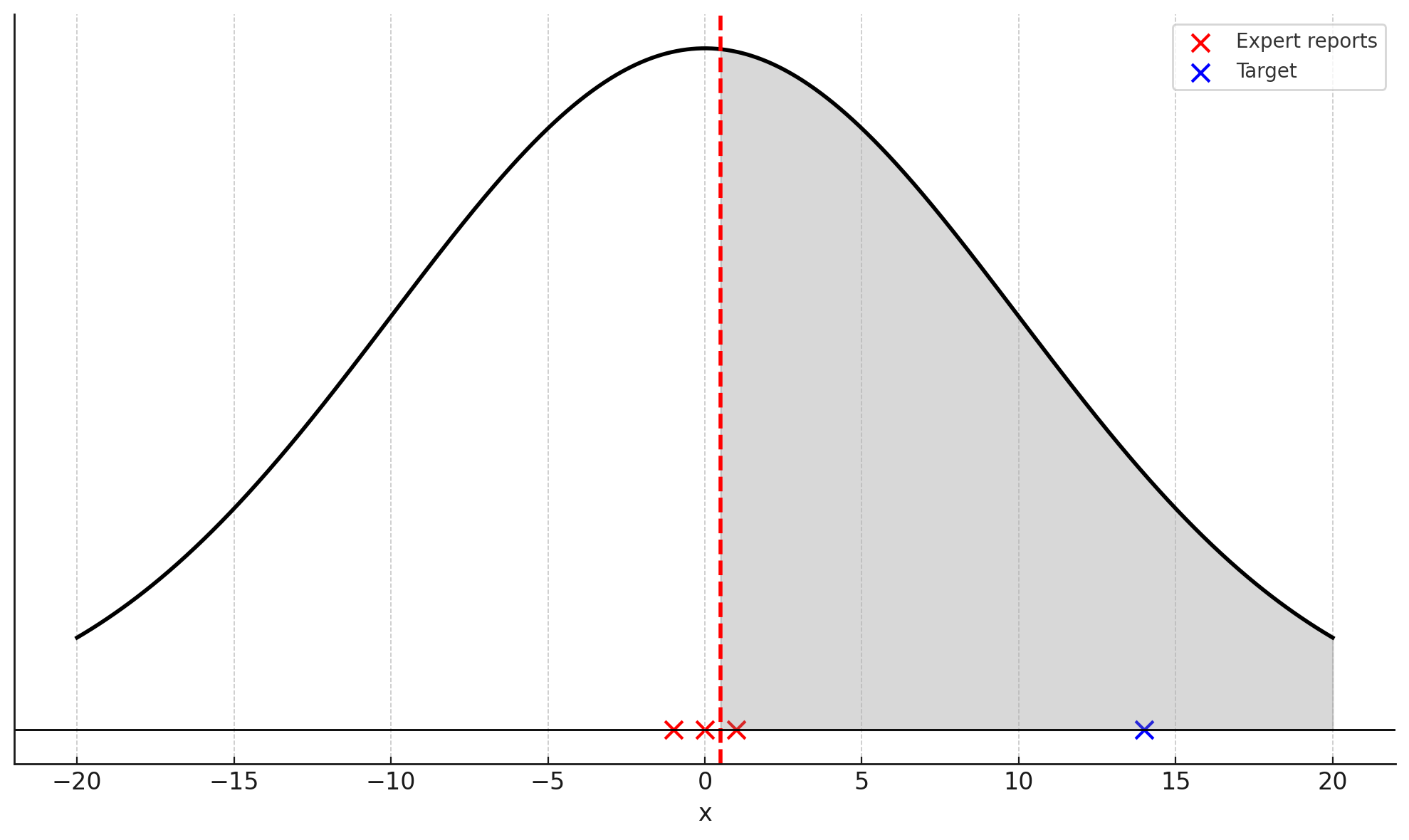}
  \caption{
    \textbf{Illustration of incentive issues in a single-question competition.}
    The $x$ axis is the report and reference space.
    All experts' \textcolor{red}{predictions (red)} are drawn from $\mathcal{N}(0, \tau_1^2)$ and the \textcolor{blue}{outcome (blue)} $Y \sim \mathcal{N}(0, \tau_2^2)$ is drawn from the black curve, with $\tau_2 \gg \tau_1$.  With high probability, experts' predictions are `close' to $0$, but the outcome is likely far from $0$. This has two consequences: (1) \textbf{statistical}: weaker, noisier experts are more likely to fall closer to the outcome than stronger, more precise experts. This means the winner is likely to be a weaker expert; (2) \textbf{game-theoretic}: an expert may unilaterally deviate to `outflank' the others and be closer to $Y$, so truthful reporting is not Bayes-Nash IC. Indeed, if they report truthfully, they win $\frac{1}{n}$ if the time, but outflanking the other experts on one `side' rockets the win probability to nearly 50\%. This is a Bayesian instance of the \citet{hotelling1929stability} game and has no IC equilibria.
  }
  \label{fig:gauss}
\end{figure}

\paragraph{Alternatives}
\citet{witkowski2023incentive} propose ELF, a strictly IC lottery-based mechanism for prediction competitions. In addition to not being deterministic, \citet{frongillo2021efficient} show that ELF is not efficient: for $n$ experts, ELF requires $O(n \log n)$ events. By contrast, they show the standard mechanism with truthful experts requires only $O(\log n)$ events, an exponentially better rate. As a compromise, they propose an FTRL-based mechanism that achieves the same $O(\log n)$ rate but only achieves approximate-IC (experts will report something close to their true beliefs). Fundamentally, both mechanisms achieve IC by adding noise to the scores of the standard mechanism faster than the label noise can, thus undermining the strategy of distorting true beliefs. Thus, the advantage an expert gains by increasing their variance is dwarfed by the loss in expected score, leading experts to not distort their reports too much.
However, to the best of our knowledge, neither mechanism has ever been used in a meaningful competition; competition runners still strongly prefer the standard mechanism despite its incentive issues. One important reason is that both alternative mechanisms rely on randomness: ELF uses random lotteries for each event, and FTRL intentionally adds noise to the winner selection. This motivates the need for \textbf{determinism} for a competition mechanism to be practically useful.

\paragraph{WOMAC}
Our WOMAC mechanism is a deterministic mechanism that is also \emph{more efficient} than the benchmark standard competition in practical settings.\footnote{In most practical settings, particularly with binary outcomes, the outcome is a fairly noisy view of the ground truth.} Our key insight is that with $n$ experts, we have $n$ pieces of information about the `ground truth'. Under reasonable distributional assumptions we can leverage this information to arrive at a much more precise estimate of the ground truth than possible by only using the noisy outcome. Specifically, we can learn a weighting over experts' to find the best ex-post aggregation of predictions given the observed outcomes. This is similar in spirit to the proxy scoring rule approach by \citet{10.5555/3298239.3298348}, except we still use the ground truth to determine weights over experts. \citet{lichtendahl2013wisdom} also showed how an aggregation of several experts' forecasts often outperforms any one expert's.  While previous designs see additional experts as a headache that worsens the prospects of finding the best performer, we see a silver lining: they bring additional information that \emph{reduces} noise in the reference solution that experts are scored against, so our mechanism is not as impacted by adding experts as the previous proposals. Lastly, we note that the idea of fitting a meta-learner from a collection of model predictions to the labels is called  \emph{stacking} in machine learning, an ensemble method used to improve predictive accuracy \citep{wolpert1992stacked}. We repurpose this machinery not for prediction accuracy, but to infer the best estimate of the ground truth for the competition setting.


\section{Model}

We provide a simple Bayesian model of how a set of $J$ \emph{experts} (with $|J| = n \geq 2$) generate predictions on a set of $I$ tasks (with $|I| = m$), along with attendant assumptions. Subsequently, we consider a \emph{winner-take-all} competition that elicits predictions on the $m$ tasks from $n$ experts and picks a \emph{single winner}.

Let $\theta_i \in \reals$ be the `ground truth' that stochastically generates \emph{outcome} $Y_i \in \mathcal{Y}$ on task $i$, where $\mathcal{Y}$ is the outcome space.\footnote{For example, if $\mathcal{Y} = \{0, 1\}$, $\theta_i$ can be the logistic transform of the Bernoulli parameter, so $Y_i \sim \text{Bernoulli}(\sigma(\theta_i))$ where $\sigma(\cdot)$ is the logistic sigmoid. If $\mathcal{Y} = \reals$, we could have $Y_i \sim \mathcal{N}(\theta_i, 1)$.}  Philosophically, $\theta_i$ is the `correct' or `best' prediction to make on task $i$ given the totality of available information. $Y_i$ is then a `noisy view' of this ground truth.

Let expert $j$'s prior over $\theta_{i}$ be $\mathcal{P}_j(\theta_i)$. Then, expert $j$ obtains signal $Z_{ij} \in \reals$ on task $i$, drawn from the distribution $\mathcal{Q}_j(Z_{ij}| \theta_i)$.\footnote{The better the expert, the more accurate this draw.} Expert $j$'s Bayesian posterior is $\mathcal{P}_j(\theta_i | Z_{ij})$, and the expert's Bayesian expectation of $\theta_i$ is ${X}_{ij} = \mathbb{E}[\theta_i | Z_{ij}] \in \reals$. We can also directly write $\mathcal{Q}_j({X}_{ij}|\theta_i)$ to reference the distribution that expert $j$'s Bayesian mean prediction is drawn from. Experts' signals and predictions are thus themselves \emph{also} noisy views of the ground truth $\theta_i$.

\paragraph{Notation} We use $i$ to index task and $j$ to index expert. We use vector notation to denote a collection of ground truths, predictions, etc. (e.g., $\vec{\theta} \in \reals^{m}$ is the vector of `ground truths').  We use boldface to denote a matrix collection across tasks and experts (e.g., ${\bm X}$ is a $m \times n$ matrix, $x_{ij}$ is expert $j$'s prediction on task $i$). Subscript $-j$ denotes a matrix that excludes a column corresponding to expert $j$, and subscript $-ij$ denotes the matrix excluding expert $j$'s column and task $i$'s row. We drop expert indices from distributions (e.,g write $\mathcal{P}_j(\vec{\theta}|\vec{Z}_j)$ as $\mathcal{P}(\vec{\theta}|\vec{Z}_j)$) when it is clear from context.

\subsection{Assumptions}
Now, we enumerate some standard assumptions about the structure of tasks and experts' predictions. We discuss these assumptions in greater detail in the Appendix.

\begin{assumption}[Tasks are IID]\label{as:iid}
  We assume that tasks are independent and identically distributed. In other words, $\theta_i \ind \theta_{i'}$ and $\mathcal{P}_j(\theta_i) = \mathcal{P}_j(\theta_{i'})$ for tasks $i \neq i'$.
\end{assumption}

\begin{assumption}[Conditional Independence]\label{as:cid}
  We assume that experts' signals and Bayesian predictions on a given task are independent given the ground truth, i.e., $Z_{ij} \ind Z_{ij'} | \theta_i$ and $X_{ij} \ind X_{ij'} | \theta_i$.
\end{assumption}

\begin{assumption}[Full Support] \label{as:support}
  We assume that experts' prior and prediction distributions have full support, i.e., $\mathcal{P}_j({\theta}_i) > 0$ and $\mathcal{Q}_j({{x}}_{ij}|{\theta}_i) > 0$ for all ${x}_{ij} \in \reals$ and $\theta_i \in \reals$, for all $i\in I$ and $j\in J$.
\end{assumption}

\section{Mechanism}

A \emph{prediction competition} is any mechanism that elicits expert predictions on some tasks and selects a single winning expert. Standard prediction competitions, which are deterministic, select the winner by scoring expert predictions $W_{ij}$ against outcomes $Y_i$, typically using mean squared error (MSE). To analyze incentives, we will consider an idealized `oracular' prediction competition that scores expert predictions using MSE directly against the ground truth $\theta_i$. We formulate our WOMAC mechanism in a subsequent section building on the insights from this analysis.

\subsection{Reports and Strategy} Let $\mathcal{R}$ be a task's report and reference solution space.\footnote{For continuous outcomes, $\mathcal{R} = \reals$. For discrete outcomes, $\mathcal{R} = [0, 1]$ if elicited as probabilities, or $\mathcal{R} = \reals$ if elicited as log odds.} Expert $j$ employs strategy $s: \reals^m \rightarrow \mathcal{R}$ to strategically report their true prediction $\vec{x}_j$ on the $m$ tasks as $\vec{w}_j = s(\vec{x}_j)$. We denote the $m \times n$ matrix of expert reports as ${\bm W}$. In general, we denote the \emph{reference solution} that expert $j$ is scored against on task $i$ as $t_{ij} \in \mathcal{R}$.   

\subsection{Competition Mechanism}

Now, we define deterministic prediction competition mechanisms, focusing on the standard \emph{outcome-based} competitions and an idealized \emph{oracle-based} competition to analyze.

\begin{definition}[Deterministic Prediction Competition $\mathcal{M}$]
  A deterministic prediction competition mechanism $\mathcal{M}$ is a function mapping $n$ experts' reported predictions $\vec{w}_1, \ldots, \vec{w}_n \in \mathcal{R}^m$ and reference solutions $\vec{t}_1, \ldots, \vec{t}_n \in \mathcal{R}^m$ to a single winner, i.e., $\mathcal{M}: \mathcal{R}^{m \times 2n}  \rightarrow \{1, \ldots, n\}$, with ties broken arbitrarily.\footnote{Since the mechanism must be deterministic in the experts' predictions, the reference solutions ${\bm T}$ may not depend stochastically on the predictions ${\bm W}$.} Expert $j$'s competition outcome under the mechanism is $\mathcal{M}_j: \mathcal{R}^{m \times 2n}  \rightarrow \{0, 1\}$.
\end{definition}

We specifically work with the most common instantiation of deterministic prediction competitions: \emph{MSE-based} prediction competitions, where $\mathcal{M}(\vec{w}_1, \ldots, \vec{w}_n, \vec{t}_1, \ldots, \vec{t}_n) = \text{argmin}_{j \in J} \|\vec{w}_j - \vec{t}\|^2$
and $\mathcal{M}_j = \mathds{1}\left[j = \text{argmin}_{j' \in J} \|\vec{w}_{j'} - \vec{t}\|^2\right]$. Additionally, in our theoretical analysis, all experts will have the same reference solution (either outcome $t_{ij} = Y_i$ or ground truth $t_{ij}=\theta_i$), so we denote the reference solution simply as $t_i \in \mathcal{R}$ for all experts, and consider $\mathcal{M}: \mathcal{R}^{m \times (n+1)}  \rightarrow \{1, \ldots, n\}$ and $\mathcal{M}_j: \mathcal{R}^{m \times (n+1)}  \rightarrow \{0, 1\}$.

\begin{definition}[Standard Prediction Competition $\mathcal{M}^Y$]
  The standard outcome-based competition $\mathcal{M}^Y$ is a deterministic MSE-based prediction competition that uses $t_{i} = y_i$, where $y_i$ is the observed outcome on task $i$.
\end{definition}

\begin{definition}[Oracular Prediction Competition $\mathcal{M}^O$]\label{def:oracle}
  An oracle-based competition $\mathcal{M}^O$ is a deterministic MSE-based prediction competition that uses $t_{i} = \theta_i$, where $\theta_i$ is the ground truth on task $i$.
\end{definition}

\subsection{Utility} The key feature of a competition is that experts tend to care more about their relative performance (e.g., position on a leaderboard) than their absolute performance. A natural way to capture this utility function is that expert $j$'s utility is 1 if they win and 0 if they lose:
\begin{equation}
  U_j(\vec{w}_j, {\bm W}_{-j}, {\bm T}) = \mathcal{M}_j({\bm W}, {\bm T})
\end{equation}
Consequently, the expected utility given reported prediction $\vec{w}_j$ and true prediction $\vec{x}_j$ is the probability of winning, i.e., $\hat{U}(\vec{w}_j, \vec{x}_j) = \mathbb{E}_{{\bm W}_{-j}, \vec{t}|\vec{x}_j}[U_j(\vec{w}_j, {\bm W}_{-j}, \vec{t})] = \text{Pr}(\mathcal{M}_j = 1 | \vec{x}_j)$. Observe that experts' utilities depend on their own report, other experts' reports, as well as the reference solution.

\section{Theoretical Results}

\subsection{Relevant Concepts}
We now cover some essential concepts used in this work. The application of these may not be immediately clear, so we invite the reader to return to this section as they encounter the concepts later in the paper.

First, we make a simple observation: if expert $j$'s predictions $\vec{x}_j$ end up closer to the reference solution $\vec{t}$, it can only help their competition outcome.

\begin{observation}[Monotonicity of $\mathcal{M}$]
  For any expert $j$, if $\|\vec{w}_j - \vec{t}\|^2 < \|\vec{w}'_j - \vec{t}\|^2$, then $\mathcal{M}_j(\vec{w}_j,  \vec{\bm W}_{-j}, \vec{t}) \ge \mathcal{M}_j(\vec{w}_j', \vec{\bm W}_{-j}, \vec{t})$.
\end{observation}
Next, we provide additional properties we will consider.
\begin{definition}[Translation-Invariance]
  We say a function $f(\vec{v}_1, \ldots, \vec{v}_k)$ is translation-invariant if $f(\vec{v}_1, \ldots, \vec{v}_k) = f(\vec{v}_1 + \vec{c}, \ldots, \vec{v}_k + \vec{c})$ for any $\vec{v}_1, \ldots, \vec{v}_k, \vec{c} \in \reals^d$.
\end{definition}

This is equivalently a location-family property for a family of conditional probability densities $\{\mathcal{Q}(\vec{x}|\vec{\theta})\}_{\vec{\theta} \in \reals^m}$ for which there is some fixed probability density function $h: \reals^m \rightarrow \reals$ such that $\mathcal{Q}(\vec{x}|\vec{\theta}) = h(\vec{x} - \vec{\theta})$.

\begin{definition}[Radially Symmetric]
  We say a function $f_{\vec{c}}: \reals^{d \times k} \to \reals$ is radially symmetric at $\vec{c}$ if there is some measurable function $h: \reals^{k} \rightarrow \reals$ such that $f_{\vec{c}}(\vec{v}_1, \ldots, \vec{v}_k) = h(\|\vec{v}_1- \vec{c}\|^2, \ldots, \|\vec{v}_k- \vec{c}\|^2)$ for any $\vec{v}_1, \ldots, \vec{v}_k \in \mathbb{R}^d$, i.e., rotation-invariant in each argument.
\end{definition}

\begin{observation}[MSE-based Prediction Competition Symmetries] $\mathcal{M}_j$ is translation-invariant and radially symmetric  at the reference solution $\vec{t}$.
\end{observation}

\begin{definition}[Strictly Radially Decreasing Function]
  We say a function $f_{\vec{c}}: \reals^{d} \to \reals$ is strictly radially decreasing from $\vec{c}$ if for any $\vec{x}, \vec{y}$ where $\|\vec{x} - \vec{c}\|^2 < \|\vec{y} - \vec{c}\|^2$, we have $f_{\vec{c}}(\vec{x}) > f_{\vec{c}}(\vec{y})$
\end{definition}

The location family, radial symmetry, and strict radial decrease apply to a wide range of probability distributions, including multivariate normal distributions and $t$-distributions. Next, we state the radial monotone likelihood ratio property, helpful for comparing distributions.

\begin{definition}[Radial Monotone Likelihood Ratio Property]
  Given two probability distributions $\mathcal{Q}_1, \mathcal{Q}_2$, we say that $\mathcal{Q}_1$ satisfies a weak radial monotone likelihood ratio property (RMLRP) with respect to $\mathcal{Q}_2$ at $\vec{c}$ if for any $\vec{v}_1, \vec{v}_2$ where $\|\vec{c} - \vec{v}_1\|^2 < \|\vec{c} - \vec{v}_2\|^2$, we have $\frac{\mathcal{Q}_1(\vec{v}_1)}{\mathcal{Q}_2(\vec{v}_1)} \geq \frac{\mathcal{Q}_1(\vec{v}_2)}{\mathcal{Q}_2(\vec{v}_2)}$. If the inequality is strict, the $\mathcal{Q}_1, \mathcal{Q}_2$ satisfy a strict RMLRP.
\end{definition}

This property states that the ratio of two distributions is diminishing monotonically towards the tails. Intuitively, if $\mathcal{Q}_1$ satisfies the MLRP with respect to $\mathcal{Q}_2$ at $\vec{c}$, $\mathcal{Q}_1$ is more `concentrated' than $\mathcal{Q}_2$ close to $\vec{c}$. This property is a helpful way to characterize `precise' and `noisy' distributions.

\begin{example}
  Consider two multivariate normal distributions with mean $\mathcal{N}_1(\vec{\mu}, \tau_1^2\mathbb{I})$, $\mathcal{N}_2(\vec{\mu}, \tau_2^2\mathbb{I})$. If $\tau_1 < \tau_2$, then $\mathcal{N}_1$ satisfies a strict RMLRP with respect to $\mathcal{N}_2$ at $\vec{\mu}$. In other words, the lower variance Gaussian (i.e., more concentrated at the mean) `dominates' the higher variance Gaussian in an RMLRP sense.
\end{example}

Finally, we present the equilibrium concept we use.

\begin{definition}[Bayes-Nash Equilibrium (BNE)]
  A strategy profile ${\bf s} = (s_1, \ldots, s_n)$ is a Bayes-Nash equilibrium if, for every player
  $j$ and every signal $z_{ij}$, given the beliefs about other players' types and strategies, the expected utility of strategy $s_j$ is at least as much as expected utility of playing any alternative strategy $s'_j$. That is, no player can increase their expected utility by unilaterally deviating, conditional on their signal $z_{ij}$.
\end{definition}

\subsection{Incentives under Ground Truth} \label{sec:incentives}

When faced with the examples of poor incentives (see Introduction), the natural question is whether \emph{competitions themselves} are the cause of poor incentives, or if it has to do with the specific choice of scoring against the outcomes $Y_i$. To study this, we consider the `idealized' case where the competition designer has access to the ground truth $\vec{\theta}$, which generates the outcome $\vec{Y}$. Although previous work \citep{monroe2025hedging} has questioned the incentive properties of prediction competitions, our first result shows that an \emph{oracular prediction competition} (Definition \ref{def:oracle}) is incentive-compatible for a wide range of belief structures. This insight will lead to our design of the WOMAC mechanism.

\begin{theorem} \label{thm:oracle}
  Suppose Assumptions \ref{as:iid}-\ref{as:support} hold. Then, truthful reporting ($\vec{w}_j = \vec{{x}}_j$ for every expert $j$) is a Bayes-Nash equilibrium in an oracular prediction competition $\mathcal{M}^O$ if the following conditions hold:

  \begin{enumerate}
    \item every expert's Bayesian posterior belief $\mathcal{P}(\vec{\theta}|\vec{X}_j)$ is \textbf{strictly radially decreasing} from $\vec{X}_j $;
    \item the conditional likelihood of each expert $j$'s expected posterior prediction $f_{\vec{\theta}}(\vec{{x}}_j) = \mathcal{Q}(\vec{{x}}_j | \vec{\theta})$ is \textbf{radially symmetric} at the ground truth $\vec{\theta}$;
    \item the conditional probability densities $\{\mathcal{Q}(\vec{x}|\vec{\theta})\}_{\vec{\theta} \in \reals^m}$ are a \textbf{location family}.
  \end{enumerate}
\end{theorem}

\paragraph{Intuition for Theorem \ref{thm:oracle}} This theorem states that for a wide category of signal distributions $\mathcal{Q}(\vec{x}_j|\vec{\theta})$ and Bayesian posteriors over the ground truth $\mathcal{P}(\vec{\theta}|\vec{X}_j)$, scoring experts against the ground truth and picking the expert with the best score is Bayes-Nash incentive-compatible.

A natural question is what kinds of distributions actually have these properties, and how plausible they are in practice. Note that the theorem applies to a very common model: normally distributed signals centered at $\vec{\theta}$, with a flat prior over the ground truth $\vec{\theta}$. This is a very realistic setting as prediction tasks contain significant context, so experts typically reason about the specific task at hand instead of adjusting predictions towards a competition prior. This theorem generically applies under any flat prior and Gaussian-like signal likelihood--a realistic and fairly broad range of distributions including $t$-distributions and exponential distributions.

So, what can we conclude from this result? First, a basic implication is that when the observation \emph{is} the ground truth (i.e., noiseless outcomes with $Y_i = \theta_i$), standard prediction competitions that score against observation are incentive-compatible. Note that this is only possible for continuous real-valued outcomes. Discrete outcomes are \emph{always} a noisy view of the continuous, real-valued ground truth $\theta_i \in \reals$, i.e., $Y_i \sim \text{Bernoulli}(\sigma(\theta_i))$. Second, it establishes that competitions are not \emph{intrinsically} untruthful mechanisms, but depend on \emph{how} experts are scored. In particular, our analysis shows that for discrete outcomes, scoring against continuous, real-valued ground truth is IC even though scoring against the binary outcome $Y_i$ is not.

Naturally, the competition designer is unlikely to have access to the ground truth $\theta_i$ in practice. To overcome this, our WOMAC mechanism will use $\theta_i^s$, an estimate the ground truth using the distribution of experts' predictions—information that was discarded by the standard prediction competition.

\subsection{Efficiency}

In addition to the question of incentives, an equally important consideration is whether competitions can accurately identify top performers with high probability. Our WOMAC mechanism will use a learned estimate $\theta_i^s$ of the ground truth, which we can model as being a noisy draw from the distribution $\mathcal{T}(\theta_i^s|\theta_i)$. Thus, we are also interested in how the noise in the estimator $\mathcal{T}(\vec{\theta}^s|\vec{\theta})$ affects the probability of identifying the best expert.

We can use an RMLRP notion of best expert: we say $j^*$ is the best expert if  $\frac{\mathcal{Q}_{j^*}(\vec{x}|\vec{\theta})}{\mathcal{Q}_{j}(\vec{x} |\vec{\theta})} > \frac{\mathcal{Q}_{j^*}(\vec{y}|\vec{\theta})}{\mathcal{Q}_{j}(\vec{y} |\vec{\theta})}$ for any $j \neq j^*$ with $\|\vec{x}-\vec{\theta}\|^2 < \|\vec{y}-\vec{\theta}\|^2$.\footnote{An interesting question is whether the expert with the highest probability of winning coincides with the expert with the best MSE with respect to the oracle. We conjecture that this true, but leave a rigorous investigation for future work.} Intuitively, the best expert's predictions are always more tightly concentrated around the ground truth than any other expert's. We also use an RMLRP notion of better estimator, i.e., a less noisy estimator is more concentrated near the ground truth. Then, a \emph{basic question} is if less noise in the estimate is strictly better for identifying the best expert. Happily, the answer is yes:

\begin{theorem}\label{thm:eff}
  Suppose Assumptions \ref{as:iid}-\ref{as:support} hold. Let $\{\mathcal{Q}_j(\vec{x}_j|\vec{\theta}\}_{j\in J}$ be \textbf{radially symmetric} at $\vec{\theta}$. Additionally, let $\mathcal{T}_A(\vec{\theta}^s|\vec{\theta}), \mathcal{T}_B(\vec{\theta}^s|\vec{\theta})$ be conditional distributions that satisfy the following:
    \begin{enumerate}
      \item They are \textbf{radially symmetric} at $\vec{\theta}$;
      \item They are \textbf{strictly radially decreasing} from $\vec{\theta}$;
      \item $\mathcal{T}_A$ satisfies a strict \textbf{radial monotone likelihood ratio property} with respect to $\mathcal{T}_B$ at $\vec{\theta}$.
    \end{enumerate}
    Then, the best expert $j^*$ has a strictly higher probability of winning a competition when $\vec{\theta}^s \sim \mathcal{T}_A$ than when $\vec{\theta}^s \sim \mathcal{T}_B$ under every ground truth $\vec{\theta} \in \reals^m$.
  \end{theorem}

  \paragraph{Intuition for Theorem \ref{thm:eff}} This result says that scoring against the more precise estimate of the ground truth (in an RMLRP sense) identifies the best expert with higher probability than scoring against the less precise estimate of the ground truth. The takeaway from this theorem is that competitions should always use the least noisy estimate of the ground truth possible, \emph{even in the absence of strategic agents}, in order to maximize the chances of identifying the best expert. Additionally, this lets us combine two results: (1) by Theorem \ref{thm:eff}, a more precise estimator can identify the best expert with higher probability than a less precise estimator given the same number of tasks $m$; (2) increasing the number of tasks $m$ can increase the probability of identifying the best expert \citep{frongillo2021efficient} with a noisy estimator. This implies that more precise estimators require \emph{fewer tasks} to identify the best expert with the same probability as a less precise estimator, i.e., more precise estimators are also more efficient. Thus, if we can arrive at a more precise estimate of $\theta_i$ (from $\mathcal{T}(\theta_i^s|\theta_i)$ than the observation $Y_i$, it offers a higher probability of selecting the best expert.

    \begin{corollary}[Efficiency]
      If an estimator $\vec{\theta}^s \sim \mathcal{T}$ satisfies the strict radial monotone likelihood ratio property with respect to the observed outcome's distribution $\mathcal{Q}(\vec{Y}|\vec{\theta})$, then using the estimate as reference solution $\vec{t} = \vec{\theta}^s$ can identify the best expert with higher probability with strictly fewer tasks than using the outcome as reference solution $\vec{t} = \vec{y}$.
    \end{corollary}

    \section{WOMAC: From Theory to Practice}\label{sec:womac}

    In the background section, we observed that scoring experts against a very noisy view of ground truth $\vec{\theta}$ incentivizes them to misreport their predictions. In our theoretical results, we established that scoring experts directly against the perfectly noiseless ground truth $\vec{\theta}$ ensures truthful reporting is Bayes-Nash incentive compatible. We also established that even if we score experts against  a merely less noisy estimate of $\theta_i$ relative to the observations $\vec{Y}_i$, we can identify the best expert with higher probability.

    These results point in one direction: for both incentives and efficiency of identifying top-performers, we want to score experts' predictions against the lowest variance unbiased estimate of the ground truth $\vec{\theta}$ that we can. Hence, the question: how can we obtain a more precise estimate of $\vec{\theta}$ than the available outcomes $\vec{Y}$? The core insight is that \emph{there is information in the distribution} of experts' reports, and we should use it instead of throwing it away.

    Indeed, given a `training' set of expert predictions ${\bm W}^{\text{train}}$ and outcomes $\vec{y}$, we can compute a `denoised' estimate of the ground truth; we fit a predictive denoising model $f_{\beta}(\cdot)$ with parameter $\beta$ by minimizing $\|\vec{y} - f_{\beta}({\bm W}^{\text{train}})\|^2$ and estimate $\vec{\theta}^s = f_\beta({\bm W}^{\text{test}})$ of the ground truth. We are essentially learning how to weight experts' reports on the training set, and applying those weights to experts' reports on the test set; this is similar to `meta-learning' or an ensemble method like stacking. Learning to optimally weight experts with $f_\beta$ is the core subroutine of our mechanism, so in this sense, the estimate of the ground truth is the \emph{wisdom of the most accurate crowd} (WOMAC). In most practical settings, this will be a lower variance unbiased estimate of the ground truth $\vec{\theta}$ compared to the outcomes $\vec{Y}$.

    Where do we obtain past predictions ${\bm W}^{\text{train}}$ and current predictions ${\bm W}^{\text{test}}$ in the context of a prediction competition? We propose computing a per-expert, per-task jackknifed `training set' to make efficient use of data and preserve determinism, so we use $({\bm W}^{\text{train}})_{ij} = {\bm W}_{-ij}$ and $({\bm W}^{\text{test}})_{ij} = \vec{w}_{i, -j}$. See Algorithm \ref{alg:algorithm} for the full mechanism.

    \begin{definition}[WOMAC $\mathcal{M}^{W}$]
      WOMAC is a deterministic MSE-based prediction competition that uses the per-expert per-task reference solution $t_{ij} = f_{\beta_{ij}}(\vec{w}_{i, -j})$ where $\beta_{ij} = \text{argmin}_{\beta'} \|\vec{y}_{-i} - f_{\beta'}({\bm W}_{-ij})\|^2$.
    \end{definition}

    \paragraph{Practical Details} A natural implementation of WOMAC would be to fit a linear/logistic regression model. While this works well when $m > n$, prediction competitions typically have more experts than predictions ($n > m$). In such cases, we can use marginal feature screening to select a smaller effective $n$. In our preferred implementation, experts are scored against the simple average of the peers in the top-$(k\cdot 100)$ percent (as judged against realized outcomes $\vec{y}$).

    \begin{equation}\label{eq:weights}
      (\beta_{ij})_{j'} =
      \begin{cases}
        \frac{1}{k(n-1)} & \frac{|\{\ell: \|\vec{y} - \vec{w}_\ell\|^2 < \|\vec{y} - \vec{w}_{j'}\|^2, \,\, \ell \neq j\}|}{n-1} < k\\
        0 & \frac{|\{\ell: \|\vec{y} - \vec{w}_\ell\|^2 < \|\vec{y} - \vec{w}_{j'}\|^2, \,\, \ell \neq j\}|}{n-1} \geq k
      \end{cases}
    \end{equation}

    One question may come from the WOMAC reference solution being a function of peers' predictions. Specifically, one might worry that the mechanism devolves into a game of predicting other experts' predictions, instead of reporting one's true predictions. We emphasize that such issues are unlikely to arise in our context, because we \emph{select peers based on performance against realized outcomes}--an external signal. This guards the reference solution from the influence of weaker, noisier experts. Intuitively, expert $j$ can be assured that their reference solution on task $i$ will be weighted towards only those peers who were most accurate on tasks $-i$ (and away from noisier, more biased experts); these experts should ex-ante also be generally accurate on task $i$ when tasks are independent. Thus, experts expect to be scored not against the realized outcome $Y_i$ but against the \emph{best possible aggregate} predictions given outcomes $\vec{Y}$.

    \section{Experiments}\label{sec:experiments}
    We now turn our attention to empirically evaluating the performance of WOMAC. We apply the WOMAC mechanism (applying Equation \ref{eq:weights} and tuning hyperparameter $k$ \emph{in-sample} to minimize MSE against outcomes) to score and rank participants in two judgmental forecasting competitions, where people compete to make predictions about YES/NO future events. We compare experts' performance under WOMAC with their performance under the standard prediction competition. Although the analysis of historical data cannot show WOMAC's improved incentive properties, we can analyze the statistical benefits.

    \begin{figure*}
      \centering
      \includegraphics[scale=0.65]{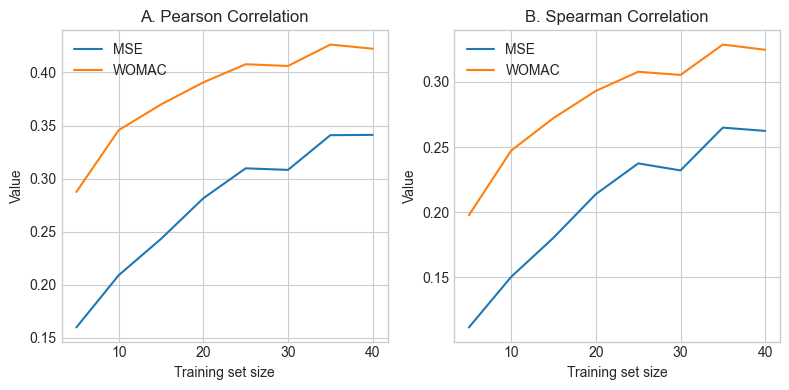}
      \caption{\textbf{[ACX Data] Plot of correlations (y-axis) between score on in-sample tasks and MSE on out-of-sample tasks, varying number of in-sample tasks (x-axis).} We construct 150 sub-samples of the full dataset of a given size (x-axis) as `training sets'. On each sub-sampled training set, we compute each expert's \textcolor{orange}{WOMAC score (orange)} and \textcolor{blue}{MSE score (blue)}, as well as the MSE on a held-out set. We compute the Pearson (L) and Spearman (R) correlations between experts' in-sample score and out-of-sample MSE score. We plot average correlation across 150 sub-samples for a given in-sample task size. The gap between the curves shows that experts' in-sample WOMAC scores are more correlated with their MSE on a held-out dataset, relative to their in-sample MSE scores, i.e., WOMAC scores are more predictive of future performance than MSE scores.}
      \label{fig:oos-acx}
    \end{figure*}

    \paragraph{Data} We analyze two datasets. First, we analyze a judgmental forecasting tournament on binary future events aimed at identifying top human forecasters, the 2023 Astral Codex Ten (ACX) Prediction Contest \citep{alexander2023acx}. In the ACX competition, we consider only forecasters who respond to all tasks, resulting in $50$ tasks and $1683$ forecasters. Second, we analyze data from the Intelligence Advanced Research Projects Activity's (IARPA's) Hybrid Forecasting Competition (HFC) program \citep{hilliard2020hfc}. We filter this data to tasks with at least $250$ responses and then filter to respondents who completed at least $50\%$ of the remaining tasks, focusing on questions with binary answers. We begin with $61$ tasks and $4243$ forecasters and end with $58$ tasks and $304$ users who met our activity threshold. We impute missing data with the mean of the realized outcomes. 

    \paragraph{Methodology} Our main analysis is to identify whether experts' `WOMAC score' on a given set of prediction tasks predicts their MSE on held-out prediction tasks \emph{better than} their `MSE score' on the same prediction tasks. An expert's `WOMAC score' is their MSE with respect to the best ex-post aggregate prediction learned through WOMAC (per Equation \ref{eq:weights}), and an expert's `MSE score' is their MSE with respect to the realized outcomes. This approach tells us how good WOMAC is at predicting experts' out-of-sample scores and ranks, relative to the standard MSE. Consistently higher correlation with WOMAC for the same amount of data would show that WOMAC is more efficient at scoring and ranking experts according to their skill. We conduct this analysis as follows:
    \begin{enumerate}
      \item For each $m_{\text{train}} \in \{5, 10, 15, \ldots, 40\}$, we take the full dataset and construct 150 randomly sub-sampled `training sets' of size $m_{\text{train}}$ and a corresponding randomly sub-sampled held-out `test-set' with size  $m_{\text{test}} = 10$.
      \item On each $(m_{\text{train}},\, m_{\text{test}}, d)$ sub-sampled dataset with $d=1,\ldots, 150$, we calculate each forecaster's in-sample WOMAC score and in-sample MSE score on the train set, and out-of-sample MSE score on the test set.
      \item For each $(m_{\text{train}},\, m_{\text{test}})$, compute $d$ Pearson and Spearman correlations between: (a) in-sample WOMAC scores and the corresponding out-of-sample MSE scores; (b) in-sample MSE score and the corresponding out-of-sample MSE scores.
      \item For each $(m_{\text{train}},\, m_{\text{test}})$, plot the average of $d$ Pearson and Spearman correlations.
    \end{enumerate}

    \paragraph{Results} On both the ACX data (see Figure \ref{fig:oos-acx}) and HFC data (see Appendix), in-sample WOMAC scores exhibit substantially better Pearson and Spearman correlation with out-of-sample MSE scores, compared to in-sample MSE scores, for a given dataset size. This implies that WOMAC is more efficient and reliable at identifying forecaster skill, compared to the standard mechanism. We also study the impact of tuning hyperparameter $k$ (see Appendix) and find taking average of $\sim$top-$5\%$ minimizes MSE against realized outcomes, though the exact choice of $k$ does not affect the performance of WOMAC very much. Finally, we also examine the impact of the number of forecasters $n$, which affects the quality of the reference solution computed from the top-$k$ (see Appendix). We consider random subsets of $100$, $400$, $1000$ and all $1683$ forecasters.  We find that the number of competitors does not make a substantial difference; this is not surprising since the average of the top-$k$ is likely already fairly good even with fewer competitors.

    \bibliography{aaai2026}


    \appendix

    \onecolumn

    \section{Assumptions}\label{sec:assump}

    We briefly discuss our assumptions:

    \begin{enumerate}
      \item \textbf{Tasks are IID}: This assumption states that tasks a priori similar and knowledge of the ground truth on one task says nothing about the ground truth on another task, and lets us reason about tasks independently. Depending on the context, this may or may not hold exactly in practice. However, there is a simple solution: when explicit clusters are available, we can fit a separate WOMAC model for each cluster of tasks (where this condition does hold). If explicit clusters are not available, matrix factorization methods can be used to reveal the underlying cluster structure. Thus, the simplifying assumption is primarily made for ease of analysis, and there are straightforward practical solutions when it doesn't strictly hold.
      \item \textbf{Conditional Independence}: This assumption states that given the ground truth, knowing one expert's prediction is not informative about another expert's prediction. This is a very common simplifying assumption in information elicitation literature, and essential for tractable analysis. Specifically, it allows us to write $\mathcal{Q}({\bm X}|\vec{\theta}) = \prod_j \mathcal{Q}(\vec{X}_j|\vec{\theta})$, and thus reason about each expert's prediction independently (for any fixed $\theta$).
      \item \textbf{Full Support}: This is the assumption that any belief is possible, even if it has vanishing probability. This assumption is used directly in our proofs that reason about the \emph{chance} that deviating to a different report makes a difference; if any report is possible, there is always some non-zero chance that a deviation makes a difference. This lets us cleanly handle the expected utility from any deviation.
    \end{enumerate}

    \section{Proofs}

    \textbf{Theorem \ref{thm:oracle}.}
    \emph{Suppose Assumptions \ref{as:iid}-\ref{as:support} hold. Then, truthful reporting ($\vec{w}_j = \vec{{x}}_j$ for every expert $j$) is a Bayes-Nash equilibrium in an oracular prediction competition $\mathcal{M}^O$ if the following conditions hold:}

    \begin{enumerate}
      \item \emph{every expert's Bayesian posterior belief $\mathcal{P}(\vec{\theta}|\vec{X}_j)$ is \textbf{strictly radially decreasing} from $\vec{X}_j $;}
      \item \emph{the conditional likelihood of each expert $j$'s expected posterior prediction $f_{\vec{\theta}}(\vec{{x}}_j) = \mathcal{Q}(\vec{{x}}_j | \vec{\theta})$ is \textbf{radially symmetric} at the ground truth $\vec{\theta}$;}
      \item \emph{the conditional probability densities $\{\mathcal{Q}(\vec{x}|\vec{\theta})\}_{\vec{\theta} \in \reals^m}$ are a \textbf{location family}.}
    \end{enumerate}

    \begin{proof}
      Let ${\bm {X}}_{-j}$ be the true reports of experts other than expert $j$ and expert $j$'s  signal is $\vec{{z}}_j$.  If expert $j$ reports $\vec{w}_j = s(\vec{x}_j)$, their expected utility $\hat{U}_j(\vec{w}_j, \vec{z}_j)$ is:
      \begin{equation}
        \hat{U}_j(\vec{w}_j, \vec{z}_j) =  \mathbb{E}_{{\bm {X}}_{-j}, \vec{\theta}|\vec{x}_j}[U_j(\vec{w}_j, {\bm {X}}_{-j}, \vec{\theta})]= \int \left[\int  \mathcal{M}_j(\vec{w}_j, {\bm {X}}_{-j}, \vec{\theta})  \mathcal{Q}({\bm {X}}_{-j}|\vec{\theta}) ~\dd{\bm {X}}_{-j}\right]  \mathcal{P}(\vec{\theta}|\vec{z}_j) ~\dd \vec{\theta}.
      \end{equation}
      Consider any two reports $\vec{u}$ and $\vec{v}$ where $0 \leq \|\vec{u} - \vec{{x}}_{j} \|^2 < \|\vec{v} - \vec{{x}}_{j} \|^2$, i.e., $\vec{u}$ is closer to the true prediction than $\vec{v}$ is. We will show that $\hat{U}(\vec{u}, \vec{z}_j) > \hat{U}(\vec{v}, \vec{z}_j)$.  This will imply that $\vec{w}_j  = \vec{{x}}_j$ maximizes expected utility. We consider the difference and will show that it is strictly positive:

      \begin{equation}\label{eq:diffexp}
        \begin{split}
          \hat{U}(\vec{u}, \vec{z}_j) - \hat{U}(\vec{v}, \vec{z}_j)  &= \int \left[\int \left[\mathcal{M}_j(\vec{u}, {\bm {X}}_{-j}, \vec{\theta}) - \mathcal{M}_j(\vec{v}, {\bm {X}}_{-j}, \vec{\theta})\right]  \mathcal{Q}({\bm {X}}_{-j}|\vec{\theta})  ~\dd {\bm {X}}_{-j} \right]  \mathcal{P}(\vec{\theta}|\vec{z}_j)  ~\dd \vec{\theta} \\
          &= \int_{\vec{\theta} \in \reals^m} g(\vec{u}, \vec{v}, \vec{\theta}) ~  \mathcal{P}(\vec{\theta}|\vec{z}_j)  ~\dd \vec{\theta}
        \end{split}
      \end{equation}
      where $g(\vec{u}, \vec{v}, \vec{\theta}) = \mathbb{E}_{{\bm {X}}_{-j}|\vec{\theta}} \left[\mathcal{M}_j(\vec{u}, {\bm {X}}_{-j}, \vec{\theta}) - \mathcal{M}_j(\vec{v}, {\bm {X}}_{-j}, \vec{\theta})\right] $. Intuitively, this is the gain in win probability by going from reporting $\vec{v}$ to reporting $\vec{u}$ when the ground truth is $\vec{\theta}$.

      Now, our strategy will be: (1) observe that we can split $\mathbb{R}^m$ into two half-spaces, where $g(\vec{u}, \vec{v}, \vec{\theta})$ is positive in one of them and negative in the other; (2) match each $g(\vec{u}, \vec{v}, \vec{\theta})$ in the first half-space with an equal and opposite-signed $g(\vec{u}, \vec{v}, \vec{\theta})$ in the other half-space; (3) argue that the weight $\mathcal{P}(\vec{\theta}|\vec{z}_j)$ on the positive signed $g(\vec{u}, \vec{v}, \vec{\theta})$ is strictly greater than the weight on the negative signed $g(\vec{u}, \vec{v}, \vec{\theta})$, so the full integral in Equation \ref{eq:diffexp} is strictly positive.

      \paragraph{Step 1: Separating opposite-signed $g(\vec{u}, \vec{v}, \vec{\theta})$ across half-spaces} Consider the half-space of points closer to $\vec{u}$ than it is to $\vec{v}$:
      \begin{equation}
        H = \big\{ \vec{\theta}': \|\vec{\theta}' - \vec{u}\|^2 \leq \|\vec{\theta}' - \vec{v}\|^2 \big\}.
      \end{equation}

      Let the complement be $H^c$ and $\partial H$ be the hyperplane separating the half-spaces. First, observe that when $\vec{\theta} \in H$, i.e., it is closer to $\vec{u}$ than to $\vec{v}$, so by \textbf{monotonicity of the mechanism}, either $\vec{v}$ wins and so does $\vec{u}$, or $\vec{v}$ loses but $\vec{u}$ wins, or both $\vec{v}$ and $\vec{u}$ lose the competition. In other words, $\mathcal{M}_j(\vec{u}, {\bm {X}}_{-j}, \vec{\theta}) - \mathcal{M}_j(\vec{v}, {\bm {X}}_{-j}, \vec{\theta}) \in \{0, 1\}$, and so $g(\vec{u}, \vec{v}, \vec{\theta}) \geq 0$ when $\vec{\theta} \in H$. By a similar argument we have $\mathcal{M}_j(\vec{u}, {\bm {X}}_{-j}, \vec{\theta}) - \mathcal{M}_j(\vec{v}, \vec{\bm w}_{-j}, \vec{\theta}) \in \{-1, 0\}$ and so $g(\vec{u}, \vec{v}, \vec{\theta}) \leq 0$ when $\vec{\theta} \in H^c$.

      We can go further and obtain strict inequalities. First, some definitions: (1) for a given $\vec{a}, \vec{b} \in \mathbb{R}^m$ and expert $j$,  let $d(\vec{a}, \vec{b}) = \|\vec{a} - \vec{b}\|^2$; (2) let $R(\vec{a}, \vec{b}, \vec{s})= \{ \vec{c}: d(\vec{a}, \vec{s}) < d(\vec{c}, \vec{s}) < d(\vec{b}, \vec{s})\}$ be the spherical shell centered at $\vec{\theta}$ of points closer to $\vec{s}$ than $\vec{v}$ but farther than $\vec{u}$.

      Now, let $\vec{{x}}^*_{-j} = \text{argmin}_{ j'\neq j} \|\vec{\theta} - \vec{{x}}_{j'}\|^2$ be the report among other experts that is closest to ground truth $\vec{\theta}$. Then, observe that when $\vec{\theta} \in H$, we have $\mathcal{M}_j(\vec{u}, {\bm {X}}_{-j}, \vec{\theta}) - \mathcal{M}_j(\vec{v}, {\bm {X}}_{-j}, \vec{\theta}) = \mathds{1}\left[{\vec{{x}}^*_{-j} \in R(\vec{u}, \vec{v}, \vec{\theta})}\right]$. In other words, expert $j$ wins by reporting $\vec{u}$ and loses by reporting $\vec{v}$ when peers' best true report $\vec{{x}}_{-j}^*$ lies inside the spherical shell $R(\vec{u}, \vec{v}, \vec{\theta})$, and so is closer than $\vec{v}$ but farther than $\vec{u}$ from the oracle-provided ground truth $\vec{\theta}$. Otherwise, reporting $\vec{u}$ or $\vec{v}$ produces the same competition outcome for expert $j$. Thus,
      \begin{equation}
        g(\vec{u}, \vec{v}, \vec{\theta}) =  \mathbb{E}_{{\bm {X}}_{-j}|\vec{\theta}} \left[\mathds{1}\left[{\vec{{x}}^*_{-j} \in R(\vec{u}, \vec{v}, \vec{\theta})}\right]\right] = \text{Pr}\left(\vec{{x}}^*_{-j} \in R(\vec{u}, \vec{v}, \vec{\theta}) \left|\right. \vec{\theta} \right).
      \end{equation}

      Since $\mathcal{Q}_j(\vec{{x}}_{j}|\vec{\theta}) > 0$ for all $j$, we have that $\text{Pr}\left(\vec{{x}}^*_{-j} \in R(\vec{u}, \vec{v}, \vec{\theta}) \left|\right. \vec{\theta} \right) > 0 \Rightarrow g(\vec{u}, \vec{v}, \vec{\theta}) > 0$ when $\vec{\theta} \in H$. By a symmetric argument, we have $g(\vec{u}, \vec{v}, \vec{\theta}) < 0$ when $\vec{\theta} \in H^c$. Thus,

      \begin{equation}\label{eq:gsign}
        g(\vec{u}, \vec{v}, \vec{\theta}) \text{ is }
        \begin{cases}
          > 0 & \vec{\theta} \in H \\
          < 0 & \vec{\theta} \in H^c.
        \end{cases}
      \end{equation}

      \paragraph{Step 2: Match $g(\vec{u}, \vec{v}, \vec{\theta})$ across half-spaces}

      Let $\vec{\theta}' \in H^c$ be the reflection of $\vec{\theta} \in H$ and ${\bm {X}}'_{-j}$ be the (column-wise) reflection of ${\bm {X}}_{-j}$ across the hyperplane $\partial H$ separating the half-spaces. Observe that by \textbf{radial symmetry and location family properties of the likelihood} at $\vec{\theta}$, we have $\mathcal{Q}({{\bm X}}'_{-j} | \vec{\theta}') = \mathcal{Q}({{\bm X}}_{-j} | \vec{\theta})$. By \textbf{radial symmetry of the mechanism} at the reference solution $\vec{\theta}$, we have $\mathcal{M}_j(\vec{u}, {\bm {X}}_{-j}, \vec{\theta}) = \mathcal{M}_j(\vec{v}, {\bm {X}}'_{-j}, \vec{\theta}')$. Thus, for any $\vec{\theta}' \in H^c$ (with mirror $\vec{\theta} \in H$), we can write:
      \begin{equation}
        \begin{split}
          g(\vec{u}, \vec{v}, \vec{\theta}') &= \int \left[\mathcal{M}_j(\vec{u}, {\bm {X}}'_{-j}, \vec{\theta}') - \mathcal{M}_j(\vec{v}, {\bm {X}}'_{-j}, \vec{\theta}')\right] \mathcal{Q}({\bm {X}}'_{-j}|\vec{\theta}') ~\dd{\bm {X}}'_{-j} \\
          &= \int \left[\mathcal{M}_j(\vec{v}, {\bm {X}}_{-j}, \vec{\theta}) - \mathcal{M}_j(\vec{u}, {\bm {X}}_{-j}, \vec{\theta})\right] \mathcal{Q}({\bm {X}}_{-j}|\vec{\theta}) ~\dd{\bm {X}}_{-j} \\
          &= -g(\vec{u}, \vec{v}, \vec{\theta}).
        \end{split}
      \end{equation}
      In other words, for any $g(\vec{u}, \vec{v}, \vec{\theta})$ with $\vec{\theta} \in H$, we have an equal and opposite $g(\vec{u}, \vec{v}, \vec{\theta}')$ with $\vec{\theta}' \in H^c$.

      \paragraph{Step 3: Positive signed $g(\vec{u}, \vec{v}, \vec{\theta})$ has greater mass}
      We can split the integral in Equation \ref{eq:diffexp} as:
      \begin{equation}\label{eq:diffexp2}
        \hat{U}(\vec{u}, \vec{z}_j) - \hat{U}(\vec{v}, \vec{z}_j) = \int_{\vec{\theta} \in H} g(\vec{u}, \vec{v}, \vec{\theta}) ~  \mathcal{P}(\vec{\theta}|\vec{z}_j)  ~\dd \vec{\theta} + \int_{\vec{\theta}' \in H^c} g(\vec{u}, \vec{v}, \vec{\theta}') ~  \mathcal{P}(\vec{\theta}'|\vec{z}_j)  ~\dd \vec{\theta}'
      \end{equation}
      Let $\vec{q}$ be the midpoint of $\vec{u}, \vec{v}$. We can reparameterize $\vec{\theta} = \vec{q} + \vec{\xi}$, with the $\partial H$-mirrored $\vec{\theta}' = \vec{q} - \vec{\xi}$. Then, collecting the integral in Equation \ref{eq:diffexp2}, we reparameterize and group like terms while using the matching property from the previous step to flips signs:
      \begin{equation}\small
        \begin{split}
          \hat{U}(\vec{u}, \vec{z}_j) - \hat{U}(\vec{v}, \vec{z}_j) &= \int_{\vec{\theta} \in H} g(\vec{u}, \vec{v}, \vec{\theta}) ~  \mathcal{P}(\vec{\theta}|\vec{z}_j)  ~\dd \vec{\theta} + \int_{\vec{\theta}' \in H^c} g(\vec{u}, \vec{v}, \vec{\theta}') ~  \mathcal{P}(\vec{\theta}'|\vec{z}_j)  ~\dd \vec{\theta}'  \\
          &= \int_{\vec{\xi}: \vec{q} + \vec{\xi} \in H} g(\vec{u}, \vec{v}, \vec{q} + \vec{\xi}) ~  \mathcal{P}(\vec{q} + \vec{\xi}|\vec{z}_j)  ~\dd \vec{\xi} + \int_{\vec{\xi}: \vec{q} + \vec{\xi} \in H} g(\vec{u}, \vec{v}, \vec{q} - \vec{\xi}) ~  \mathcal{P}(\vec{q} - \vec{\xi}|\vec{z}_j)  ~\dd \vec{\xi}  \\
          &= \int_{\vec{\xi}: \vec{q} + \vec{\xi} \in H} \left[ g(\vec{u}, \vec{v}, \vec{q} + \vec{\xi}) ~  \mathcal{P}(\vec{q} + \vec{\xi}|\vec{z}_j)    -g(\vec{u}, \vec{v}, \vec{q} + \vec{\xi}) ~  \mathcal{P}(\vec{q} - \vec{\xi}|\vec{z}_j)\right]  ~\dd \vec{\xi}  \\
          &= \int_{\vec{\theta} \in H} g(\vec{u}, \vec{v}, \vec{\theta}) ~  \left[\mathcal{P}(\vec{\theta}|\vec{z}_j) -  \mathcal{P}(\vec{\theta}'|\vec{z}_j) \right] ~\dd \vec{\theta}
        \end{split}
      \end{equation}

      Since the \textbf{posterior is strictly radially decreasing}, and the reflection $\vec{\theta}'$ is farther from $\vec{{x}}_j$ than $\vec{\theta}$ is, we have $\mathcal{P}(\vec{\theta}|\vec{z}_j) -  \mathcal{P}(\vec{\theta}'|\vec{z}_j) > 0$. From Equation \ref{eq:gsign}, we have that $g(\vec{u}, \vec{v}, \vec{\theta}) > 0$ for any $\vec{\theta} \in H$. Thus, the integral above is strictly positive, and $\hat{U}(\vec{u}, \vec{z}_j) - \hat{U}(\vec{v}, \vec{z}_j) > 0$.

      Since this holds for any $\vec{u}, \vec{v}$ where $d(\vec{u}, \vec{{x}}_j) < d(\vec{v}, \vec{{x}}_j)$, we have that expected utility is maximized when expert $j$ reports their true prediction $\vec{w}_j =\vec{{x}}_j$ when all other experts also report truthfully, i.e., truthful reporting is a Bayes-Nash Equilibrium.
    \end{proof}

    \noindent{\textbf{Theorem \ref{thm:eff}}}
    \emph{Suppose Assumptions \ref{as:iid}-\ref{as:support} hold. Let $\{\mathcal{Q}_j(\vec{x}_j|\vec{\theta}\}_{j\in J}$ be \textbf{radially symmetric} at $\vec{\theta}$. Additionally, let $\mathcal{T}_A(\vec{\theta}^s|\vec{\theta}), \mathcal{T}_B(\vec{\theta}^s|\vec{\theta})$ be conditional distributions that satisfy the following:
        \begin{enumerate}
          \item They are \textbf{radially symmetric} at $\vec{\theta}$;
          \item They are \textbf{strictly radially decreasing} from $\vec{\theta}$;
          \item $\mathcal{T}_A$ satisfies a strict \textbf{radial monotone likelihood ratio property} with respect to $\mathcal{T}_B$ at $\vec{\theta}$.
        \end{enumerate}
        Then, the best expert $j^*$ has a strictly higher probability of winning a competition when $\vec{\theta}^s \sim \mathcal{T}_A$ than when $\vec{\theta}^s \sim \mathcal{T}_B$ under every ground truth $\vec{\theta} \in \reals^m$.
      } \\

      Before we prove this theorem, we prove the following useful lemma:

      \begin{lemma}\label{lem:sym}
        Suppose $\{\mathcal{Q}_j(\vec{x}_j|\vec{\theta}\}$ is \textbf{radially symmetric} at $\vec{\theta}$. Let
          \begin{equation}
            f_{j^*}(\vec{\theta}^{s}, \vec{\theta}) = \int \mathcal{M}_j({\vec{x}}_{j^*}, {\bm {X}}_{-j^*}, \vec{\theta}^{s}) ~\mathcal{Q}({\bm {X}}|\vec{\theta}) ~\dd{\bm {X}}
          \end{equation}
          where $\mathcal{Q}({\bf X}|\vec{\theta}) = \prod_{j} \mathcal{Q}_j(\vec{x}_j|\vec{\theta})$. Then $f_{j^*}(\cdot, \vec{\theta})$ is radially symmetric about and strictly radially decreasing from $\vec{\theta}$.
        \end{lemma}

        \begin{proof}[Proof of Lemma \ref{lem:sym}]
          This is essentially the claim that for a given $\vec{\theta}$ and reference solution $\vec{\theta}^s$, expert $j^*$'s ex-ante win probability formed by averaging over their own predictions and peer experts' predictions is peaked at $\vec{\theta}$ and symmetrically decreases radially outward from $\vec{\theta}$ (when all experts are truthful). To begin, we can write:
          \begin{equation}\label{eq:lem1}
            \begin{split}
              f_{j^*}(\vec{\theta}^{s}, \vec{\theta}) &= \int \mathcal{M}_{j^*}(\vec{{x}}_{j^*}, {\bm {X}}_{-j^*}, \vec{\theta}^{s}) \mathcal{Q}({\bm {X}}|\vec{\theta}) ~\dd{\bm {X}} \\
              &= \int \left[\int \mathcal{M}_{j^*}(\vec{{x}}_{j^*}, {\bm {X}}_{-j^*}, \vec{\theta}^{s}) \mathcal{Q}({\bm {X}}_{-j^*}|\vec{\theta}) ~\dd {\bm X}_{-j^*}\right] \mathcal{Q}_{j^*}(\vec{x}_{j^*}|\vec{\theta})~\dd\vec{{x}}_{j^*} \\
              &= \int h(\vec{x}_{j^*}, \vec{\theta}^s, \vec{\theta}) ~\mathcal{Q}_{j^*}({ \vec{x}}_{j^*}|\vec{\theta}) \dd\vec{{x}}_{j^*}
            \end{split}
          \end{equation}
          where $h(\vec{x}_{j^*}, \vec{\theta}^s, \vec{\theta}) = \int \mathcal{M}_{j^*}(\vec{{x}}_{j^*}, {\bm {X}}_{-j^*}, \vec{\theta}^{s}) ~\mathcal{Q}({\bm {X}}_{-j^*}|\vec{\theta}) ~\dd {\bm X}_{-j^*}$ is the probability that expert $j^*$ wins fixing their (truthful) report $\vec{x}_j$, and $\vec{\theta}^s, \vec{\theta}$.

          \paragraph{Step 1: Analyze $h(\vec{x}_{j^*}, \vec{\theta}^s, \vec{\theta})$, expert $j^*$'s probability of winning} Writing $r = \|\vec{x}_{j^*} - \vec{\theta}^s\|$ as distance between the reference solution and expert $j^*$'s report, observe:
          \begin{equation}
            \begin{split}
              h(\vec{x}_{j^*}, \vec{\theta}^s, \vec{\theta}) &=\int \mathcal{M}_{j^*}(\vec{{x}}_{j^*}, {\bm {X}}_{-j^*}, \vec{\theta}^{s}) ~\mathcal{Q}({\bm {X}}_{-j^*}|\vec{\theta}) ~\dd {\bm X}_{-j^*} \\
              &= \prod_{j \neq j^*} \left(1 - \int_{\|\vec{x} - \vec{\theta}^s \| \leq r} \mathcal{Q}_j(\vec{x} | \vec{\theta}) ~\dd \vec{x}\right) \\
              &= \prod_{j \neq j^*} h_{j}(\vec{x}_{j^*}, \vec{\theta}^s, \vec{\theta})
            \end{split}
          \end{equation}
          where $ h_{j}(\vec{x}_{j^*}, \vec{\theta}^s, \vec{\theta}) = 1 - \int_{\|\vec{x} - \vec{\theta}^s \| \leq r} \mathcal{Q}_j(\vec{x} | \vec{\theta}) ~\dd \vec{x}$ 
          is the probability that expert $j^*$ is closer to $\vec{\theta}^s$ than a given expert $j$.

          \textbf{Observe} that $h_{j}(\vec{x}_{j^*}, \vec{\theta}^s, \vec{\theta})$ is radially symmetric in $\vec{\theta}^s$ about $\vec{\theta}$ \textbf{due to the radial symmetry of} $\mathcal{Q}_j(\vec{x}|\vec{\theta})$. This gives us that $h(\vec{x}_{j^*}, \vec{\theta}^s, \vec{\theta})$ is radially symmetric in $\vec{\theta}^s$ about $\vec{\theta}$, which \textbf{in conjunction with the radial symmetry of} $\mathcal{Q}_{j^*}$ in Equation \ref{eq:lem1}, \textbf{gives us radial symmetry of $f_{j^*}(\vec{\theta}^{s}, \vec{\theta})$.} This establishes the first claim.

          Next, observe that for any $\vec{x}, \vec{x}'$ that are equidistant from $\vec{\theta}^s$, we have $h_j(\vec{x}, \vec{\theta}^s, \vec{\theta}) = h_j(\vec{x}', \vec{\theta}^s, \vec{\theta})$ and $h(\vec{x}, \vec{\theta}^s, \vec{\theta}) = h(\vec{x}', \vec{\theta}^s, \vec{\theta})$, since the balls $R_{\vec{\theta}^s}(\vec{x}) = R_{\vec{\theta}^s}(\vec{x}')$, with both $\vec{x}, \vec{x}'$ on the ball's surface. In other words, conditional on $\vec{\theta}^s, \vec{\theta}$, the win probability is constant for any report with a constant distance from $\vec{\theta}^s$. Thus, with mild abuse of notation we can write this as $h_j(r, \vec{\theta}^s, \vec{\theta})$ when convenient.

          \paragraph{Step 2: Analyze $\partial h(\vec{x}_{j^*}, \vec{\theta}^s(\xi), \vec{\theta}) / \partial {\xi}$, change in expert $j^*$'s probability of winning when $\vec{\theta}^s$ moves outward from $\vec{\theta}$ radially in direction $\xi$} Now, without loss of generality, consider $\vec{\theta}^s$ which is axis aligned on the first coordinate as $\vec{\theta}^s(\xi) = \vec{\theta} + \vec{\xi} = (\theta_1 + \xi, \theta_2 \ldots, \theta_m)$ where $\vec{\xi} = (\xi, 0, \ldots, 0)$ and $\xi > 0$, and (with further mild abuse of notation) consider how $h(\vec{x}_{j^*}, \vec{\theta}^s(\xi), \vec{\theta})$ changes as we move 
          $\vec{\theta}^s(\xi)$ away from $\vec{\theta}$
          in the direction of $\vec \xi$:
          \begin{equation}\label{eq:derivh}
            \begin{split}
              \frac{\partial}{\partial \xi} h(\vec{x}_{j^*}, \vec{\theta}^s(\xi), \vec{\theta}) &= \sum_{j \neq j^*} \left(\frac{\partial }{\partial \xi} h_{j}(\vec{x}_{j^*}, \vec{\theta}^s(\xi), \vec{\theta}) \right)\prod_{j' \neq j^*, j} h_{j'}(\vec{x}_{j^*}, \vec{\theta}^s, \vec{\theta}) \\
              &= \sum_{j \neq j^*} \left(\frac{\partial }{\partial \xi} h_j(\vec{x}_{j^*}, \vec{\theta}^s(\xi), \vec{\theta}) \right) \frac{h(\vec{x}_{j^*}, \vec{\theta}^s, \vec{\theta})}{h_{j}(\vec{x}_{j^*}, \vec{\theta}^s, \vec{\theta})} \\
              &= \sum_{j \neq j^*} \left(- \frac{\partial}{\partial \xi}\int_{\|\vec{x} - \vec{\theta} - \vec{\xi} \| \leq \|\vec{x}_{j^*} - \vec{\theta} - \vec{\xi}\|} \mathcal{Q}_j(\vec{x} |\vec{\theta}) ~\dd \vec{x} \right) \frac{h(\vec{x}_{j^*}, \vec{\theta}^s, \vec{\theta})}{h_{j}(\vec{x}_{j^*}, \vec{\theta}^s, \vec{\theta})}\\
              &= \sum_{j \neq j^*}  \left(- \int_{\|\vec{x} - \vec{\theta}^s\| = r} \mathcal{Q}_j(\vec{x} |\vec{\theta}) \left( \frac{x_1 - (x_{j^*})_1}{r} \right) ~\dd S(\vec{x}) \right) \frac{h(r, \vec{\theta}^s, \vec{\theta})}{h_{j}(r, \vec{\theta}^s, \vec{\theta})}
            \end{split}
          \end{equation}
          where the final step includes an integral of peer $j$'s reports over the surface of the sphere at radius $r$ from $\vec{\theta}^s(\xi)$. This tells us how $h(\vec{x}_{j^*}, \vec{\theta}^s, \vec{\theta})$ changes in any arbitrary outward radial direction, and we can use this to compute how the original expression $f_{j^*}(\vec{\theta}^s, \vec{\theta})$ varies when $\theta^s$ moves outward radially.

          \paragraph{Step 3: Revisit $f_{j^*}(\vec{\theta}^{s}, \vec{\theta})$ and take derivative in radial direction} Let us take the partial derivative with respect to $\xi$, and rewrite the integral over $\vec{x}_{j^*}$ into a double-integral over the points on the sphere $\vec{y}$ at a fixed radius $r$ from $\vec{\theta}^s$ and over the radius $r$. 
          \begin{equation}
            \begin{split}
              \frac{\partial}{\partial \xi} f_{j^*}(\vec{\theta}^{s}(\xi), \vec{\theta}) &= \int \frac{\partial}{\partial \xi} h(\vec{x}_{j^*}, \xi, \vec{\theta}) ~\mathcal{Q}_{j^*}({ \vec{x}}_{j^*}|\vec{\theta}) \dd\vec{{x}}_{j^*} \\
              &= \int_{r=0}^{\infty} \int_{\|\vec{y} - \vec{\theta}^s\| = r} \frac{\partial}{\partial \xi}  h(\vec{y}, \xi, \vec{\theta}) ~\mathcal{Q}_{j^*}({ \vec{y}}|\vec{\theta}) ~\dd S(\vec{{y}}) ~\dd r.
            \end{split}
          \end{equation}
          Note that the inner integral over $\vec{y}$ is over expert $j^*$'s reports over surface of the ball centered at $\vec{\theta}^s$ with radius $r$, while the integral over $\vec{x}$ in Equation \ref{eq:derivh} is over some peer $j$'s reports \emph{on the surface of the same ball}. We will leverage this to simplify the computation of the \emph{change} in $f_{j^*}(\vec{\theta}^s, \vec{\theta})$ as we move $\vec{\theta}^s$ outward radially in the direction of $\xi$:
          \begin{equation*}\small
            \begin{split}
              & \frac{\partial}{\partial \xi} f_{j^*}(\vec{\theta}^s(\xi), \vec{\theta})
              = \int_{0}^{\infty} \int_{\|\vec{y} - \vec{\theta}^s\| = r} \frac{\partial}{\partial \xi}  h(\vec{y}, \vec{\theta}^s(\xi), \vec{\theta}) ~\mathcal{Q}_{j^*}({ \vec{y}} | \vec{\theta}) ~\dd S(\vec{{y}}) ~\dd r \\
              &= \int_{0}^{\infty} \int_{\|\vec{y} - \vec{\theta}^s\| = r}  \left[\sum_{j \neq j^*} \frac{h(r, \vec{\theta}^s, \vec{\theta})}{h_{j}(r, \vec{\theta}^s, \vec{\theta})}\left(- \int_{\|\vec{x} - \vec{\theta}^s\| = r} \mathcal{Q}_j(\vec{x} |\vec{\theta}) \left( \frac{x_1 - y_1}{r} \right) ~\dd S(\vec{x}) \right) \right] ~\mathcal{Q}_{j^*}({ \vec{y}}|\vec{\theta}) ~\dd S(\vec{{y}}) ~\dd r \\
              &= \int_{0}^{\infty}   \sum_{j \neq j^*} \frac{h(r, \vec{\theta}^s,\vec{\theta})}{h_{j}(r, \vec{\theta}^s,\vec{\theta})}\left[ \int_{\|\vec{y} - \vec{\theta}^s\| = r} \left(- \int_{\|\vec{x} - \vec{\theta}^s\| = r} \mathcal{Q}_j(\vec{x} |\vec{\theta})  \left( \frac{x_1 - y_1}{r} \right) ~\dd S(\vec{x})  \right) ~\mathcal{Q}_{j^*}({ \vec{y}}|\vec{\theta})  ~\dd S(\vec{{y}}) \right]  ~\dd r \\
              &= \int_{0}^{\infty}   \sum_{j \neq j^*} \frac{h(r, \vec{\theta}^s\vec{\theta})}{h_{j}(r, \vec{\theta}^s,\vec{\theta})}\left(- \int_{\|\vec{y} - \vec{\theta}^s\| = r}\int_{\|\vec{x} - \vec{\theta}^s\| = r} \mathcal{Q}_j(\vec{x} |\vec{\theta}) ~\mathcal{Q}_{j^*}({ \vec{y}}|\vec{\theta}) \left( \frac{x_1 - y_1}{r} \right) ~\dd S(\vec{x})  ~\dd S(\vec{{y}}) \right)   ~\dd r \\
              &= \int_{0}^{\infty}   \sum_{j \neq j^*} \frac{h(r, \vec{\theta}^s,\vec{\theta})}{h_{j}(r, \vec{\theta}^s,\vec{\theta})}\left(- H(\vec{\theta}^s, \vec{\theta}, r) \right)   ~\dd r
            \end{split}
          \end{equation*}
          where $H(\vec{\theta}^s, \vec{\theta}, r) = \int_{\|\vec{y} - \vec{\theta}^s\| = r}\int_{\|\vec{x} - \vec{\theta}^s\| = r} \mathcal{Q}_j(\vec{x} | \vec{\theta}) ~\mathcal{Q}_{j^*}({ \vec{y}}|\vec{\theta}) \left( \frac{x_1 - y_1}{r} \right) ~\dd S(\vec{x})  ~\dd S(\vec{{y}})$. This double integral is over all potential expert $j^*$'s reports at some distance $r$ away from $\vec{\theta}^s$, and over a peer $j$'s reports the same distance away.

          \paragraph{Step 4: Establish that $H(\vec{\theta}^s, \vec{\theta}) > 0$} Now, adding $H(\vec{\theta}^s, \vec{\theta})$ to itself 
          and swapping the variable names $\vec x$ and $\vec y$,
          we obtain:
          \begin{equation*}\small
            \begin{split}
              H(\vec{\theta}^s) &= \frac12 \left[\int_{\|\vec{y} - \vec{\theta}^s\| = r}\int_{\|\vec{x} - \vec{\theta}^s\| = r} \mathcal{Q}_j(\vec{x} |\vec{\theta}^s) ~\mathcal{Q}_{j^*}({ \vec{y}}|\vec{\theta}) \left( \frac{x_1 - y_1}{r} \right) ~\dd S(\vec{x})  ~\dd S(\vec{{y}})  \right. \\
              &\phantom{====} + ~~\left. \int_{\|\vec{x} - \vec{\theta}^s\| = r}\int_{\|\vec{y} - \vec{\theta}^s\| = r} \mathcal{Q}_j(\vec{y} |\vec{\theta}) ~\mathcal{Q}_{j^*}({ \vec{x}}|\vec{\theta}) \left( \frac{y_1 - x_1}{r} \right) ~\dd S(\vec{y})  ~\dd S(\vec{{x}})\right] \\
              &= \frac1{2r} \left[\int_{\|\vec{y} - \vec{\theta}^s\| = r}\int_{\|\vec{x} - \vec{\theta}^s\| = r} \left(\mathcal{Q}_j(\vec{x} |\vec{\theta}) ~\mathcal{Q}_{j^*}({ \vec{y}}|\vec{\theta}) - \mathcal{Q}_j(\vec{y} |\vec{\theta}) ~\mathcal{Q}_{j^*}({ \vec{x}}|\vec{\theta}) \right) \left( x_1 - y_1 \right) ~\dd S(\vec{x})  ~\dd S(\vec{{y}})  \right].
            \end{split}
          \end{equation*}
          Now, suppose $\|\vec{y} - \vec{\theta}\| < \|\vec{x} - \vec{\theta}\|$. By the \textbf{radial MLRP of $\mathcal{Q}_j^*$ with respect to $\mathcal{Q}_j$}, we have $ \frac{\mathcal{Q}_{j^*}(\vec{y}|\vec{\theta})}{\mathcal{Q}_{j}(\vec{y}|\vec{\theta})} > \frac{\mathcal{Q}_{j^*}(\vec{x}|\vec{\theta})}{\mathcal{Q}_{j}(\vec{x}|\vec{\theta})} \Rightarrow \mathcal{Q}_j(\vec{x} |\vec{\theta}) ~\mathcal{Q}_{j^*}({ \vec{y}}|\vec{\theta}) - \mathcal{Q}_j(\vec{y} |\vec{\theta}) ~\mathcal{Q}_{j^*}({ \vec{x}}|\vec{\theta}) > 0$. Additionally,
          \begin{equation*}
            \|\vec{y} - \vec{\theta}\|^2 = 2 (\vec{\theta}^s-\vec{\theta}) \cdot (\vec{y} - \vec{\theta}) - (\|(\vec{\theta}^s-\vec{\theta}) - (\vec{y} - \vec{\theta})\|^2 - \|\vec{\theta}^s - \vec{\theta}\|^2) = 2 \xi (y_1 - \theta_1) - (r^2 - \|\vec{\theta}^s - \vec{\theta}\|^2)
          \end{equation*}
          Similarly, $\|\vec{x} - \vec{\theta}\|^2 =  2 \xi (x_1-\theta_1) - (r^2 - \|\vec{\theta}^s - \vec{\theta}\|^2)$. Thus,
          \begin{align*}
            \|\vec{y}-\vec{\theta}\|^2 < \|\vec{x}-\vec{\theta}\|^2 & \Rightarrow  2 \xi (y_1-\theta_1) - (r^2 - \|\vec{\theta}^s - \vec{\theta}\|^2) < 2 \xi (x_1 -\theta_1) - (r^2 - \|\vec{\theta}^s - \vec{\theta}\|^2) \\
            & \Rightarrow x_1 - y_1 > 0.
          \end{align*}
          Thus, we have that the integrand $\left(\mathcal{Q}_j(\vec{x} |\vec{\theta}) ~\mathcal{Q}_{j^*}({ \vec{y}}|\vec{\theta}) - \mathcal{Q}_j(\vec{y} |\vec{\theta}) ~\mathcal{Q}_{j^*}({ \vec{x}}|\vec{\theta}) \right) \left( x_1 - y_1 \right) > 0$, which gives us that $H(\vec{\theta}^s,\vec{\theta}) > 0$.

          Similarly, if $\|\vec{x}-\vec{\theta}\| < \|\vec{y}-\vec{\theta}\|$, \textbf{by radial MLRP of $\mathcal{Q}_j^*$ with respect to $\mathcal{Q}_j$} we have $\mathcal{Q}_j(\vec{x} |\vec{\theta}) ~\mathcal{Q}_{j^*}({ \vec{y}}|\vec{\theta}) - \mathcal{Q}_j(\vec{y}|\vec{\theta}) ~\mathcal{Q}_{j^*}({ \vec{x}}|\vec{\theta}) < 0$. But since this implies $x_1 < y_1$, the integrand is once again positive. Thus, we have established that $H(\vec{\theta}^s, \vec{\theta}) > 0$. This finally gives us that:

          \begin{equation}
            \frac{\partial}{\partial \xi } f_{j^*}(\vec{\theta}^s(\xi), \vec{\theta}) = \int_{0}^{\infty}   \sum_{j \neq j^*} \frac{h(r, \vec{\theta}^s,\vec{\theta})}{h_{j}(r, \vec{\theta}^s,\vec{\theta})}\left(- H(\vec{\theta}^s,\vec{\theta}) \right)   ~\dd r < 0.
          \end{equation}

          Thus, we have established that $f_{j^*}(\vec{\theta}^s, \vec{\theta})$ is strictly decreasing radially outward in $\vec{\theta}^s$.
        \end{proof}

        \begin{proof}[Proof of Theorem \ref{thm:eff}]
          Now, we can prove Theorem \ref{thm:eff}. Conditioning on ground truth $\vec{\theta} \in \reals^m$, when the estimator $\vec{\theta}^{s}$ is drawn from $\mathcal{T}_A$, expert $j^*$'s expected utility is:
          \begin{equation}
            \begin{split}
              \mathbb{E}_{\bm X, \vec{\theta}^s | \vec{\theta}}[U_{j^*} ] &= \int \left[\int \mathcal{M}_{j^*}({{x}}_{j^*}, {\bm {X}}_{-j^*}, \vec{\theta}^{s}) ~ \mathcal{Q}({\bm {X}}|\vec{\theta}) ~ \dd{\bm {X}} \right] ~\mathcal{T}_A(\vec{\theta}^{s}|\vec{\theta}) ~ \dd\vec{\theta}^{s} \\ 
              &= \int  f_{j^*}(\vec{\theta}^{s}, \vec{\theta})  ~\mathcal{T}_A(\vec{\theta}^{s}|\vec{\theta}) ~ \dd\vec{\theta}^{s}, 
            \end{split}
          \end{equation}
          where we define $f_{j^*}(\vec{\theta}^{s}, \vec{\theta}) = \int \mathcal{M}_j({{x}}_{j^*}, {\bm {X}}_{-j^*}, \vec{\theta}^{s}) \mathcal{Q}({\bm {X}}|\vec{\theta}) \dd{\bm {X}}$.
          \textbf{By Lemma \ref{lem:sym}}, $f_{j^*}(\cdot, \vec{\theta})$ is a symmetric unimodal strictly decreasing function centered at $\vec{\theta}$ since expert $j^*$ is the best expert. Then, the difference in expected utilities (i.e., win probabilities) when the reference solution is drawn from $\mathcal{T}_A$ versus when it is drawn from $\mathcal{T}_B$ is:
          \begin{equation}
            \mathbb{E}_{{\bm {X}}, \vec{\theta}^{s} \sim \mathcal{T}_A|\vec{\theta}}[U_{j^*}] - \mathbb{E}_{{\bm {X}}, \vec{\theta}^{s} \sim \mathcal{T}_B|\vec{\theta}}[U_{j^*}] = \int  f_{j^*}(\vec{\theta}^{s}, \vec{\theta})  \left[\mathcal{T}_A(\vec{\theta}^{s}|\vec{\theta}) - \mathcal{T}_B(\vec{\theta}^{s}|\vec{\theta})\right] ~\dd\vec{\theta}^{s}. 
          \end{equation}
          Let $T>0$ be the radius from $\vec{\theta}$ at which $\mathcal{T}_A$ and $\mathcal{T}_B$ intersect. \textbf{Since $\mathcal{T}_A$ satisfies a radial MLRP condition with respect to $\mathcal{T}_B$} and they are both \textbf{strictly radially decreasing} from $\vec{\theta}$, there is exactly one $T>0$ where the two distributions intersect. Writing $D(\vec{\theta}^{s}, \vec{\theta}) = \mathcal{T}_A(\vec{\theta}^{s}|\vec{\theta}) - \mathcal{T}_B(\vec{\theta}^{s}|\vec{\theta})$, unit vector as $\vec{e}$, and using the radial symmetry of $D(\cdot, \vec{\theta})$ (inherited from \textbf{the radial symmetry of $\mathcal{T}_A, \mathcal{T}_B$}):
          \begin{align}
            & \mathbb{E}_{{\bm {X}}, \vec{\theta}^{s} \sim \mathcal{T}_A|\vec{\theta}}[U_{j^*}] - \mathbb{E}_{{\bm {X}}, \vec{\theta}^{s} \sim \mathcal{T}_B|\vec{\theta}}[U_{j^*}] \nonumber \\
            &= \int_{\vec{\theta}^s \in \reals^m}   ~f_{j^*}(\vec{\theta}^s, \vec{\theta}) ~D(\vec{\theta}^s, \vec{\theta}) ~\dd \vec{\theta}^s  \nonumber \\
            &= \int_{r>0} \int_{\vec{v}: \|\vec{v}\| = 1}   ~f_{j^*}(\vec{\theta} + r\vec{v}, \vec{\theta}) ~D(\vec{\theta} + r\vec{v}, \vec{\theta}) ~\dd \vec{v} ~\dd r \nonumber \\
            &= \int_{r>0} \int_{\vec{v}: \|\vec{v}\| = 1}   ~f_{j^*}(\vec{\theta} + r\vec{e}, \vec{\theta}) ~D(\vec{\theta} + r\vec{e}, \vec{\theta}) ~\dd \vec{v} ~\dd r  \nonumber \\
            &=C \int_{r>0} f_{j^*}(\vec{\theta} + r\vec{e}, \vec{\theta}) ~D(\vec{\theta} + r\vec{e}, \vec{\theta}) ~\dd r  \nonumber \\
            &= C \int_{0 < r \leq T} f_{j^*}(\vec{\theta} + r\vec{e}, \vec{\theta})  ~D(\vec{\theta} + r\vec{e}, \vec{\theta}) ~\dd \vec{\theta}^s  ~ + ~ C \int_{r > T} f_{j^*}(\vec{\theta} + r\vec{e}, \vec{\theta})  ~D(\vec{\theta} + r\vec{e}, \vec{\theta}) ~\dd \vec{\theta}^s  \nonumber \\
            &=  C \int_{0 < r \leq T} \left(f_{j^*}(\vec{\theta} + r\vec{e}, \vec{\theta})  - f_{j^*}(\vec{\theta} + T\vec{e}, \vec{\theta})\right)  ~D(\vec{\theta} + r\vec{e}, \vec{\theta}) ~\dd \vec{\theta}^s  \label{eq:long-1}  \\
            &\phantom{=} ~~ + C \int_{r > T} \left(f_{j^*}(\vec{\theta}^s + r\vec{e}, \vec{\theta})   - f_{j^*}(\vec{\theta} + T\vec{e}, \vec{\theta}) \right)  ~D(\vec{\theta} + r\vec{e}, \vec{\theta}) ~\dd \vec{\theta}^s  \label{eq:long-2}  \\
            &\phantom{=} ~~ + C \int_{r > 0} f_{j^*}(\vec{\theta} + T\vec{e}, \vec{\theta})  ~D(\vec{\theta} + r\vec{e}, \vec{\theta}) ~\dd r    \label{eq:long-3}
          \end{align}
          where $C = \int_{\vec{v}: \|\vec{v}\|=1} ~\dd\vec{v}>0$ is the constant from integrating out the angle. Since $\mathcal{T}_A, \mathcal{T}_B$ are radially symmetric with mode at $\vec{\theta}$:
          \begin{equation*}
            \eqref{eq:long-3} = C~f_{j^*}(\vec{\theta} + T\vec{e}, \vec{\theta}) \int_{r  > 0}  D(\vec{\theta} + r\vec{e}, \vec{\theta}) ~\dd r  = 0.
          \end{equation*}

          Additionally, when $0 \leq r \leq T$, $f_{j^*}(\vec{\theta} + r\vec{e}, \vec{\theta})  - f_{j^*}(\vec{\theta} + T\vec{e}, \vec{\theta}) > 0$ (since by Lemma \ref{lem:sym} is \textbf{strictly radially decreasing}) and $D(\vec{\theta} + r\vec{e}, \vec{\theta})> 0$ (\textbf{by radial MLRP}),
          so Equation \eqref{eq:long-1} is positive.
          Similarly,  when $r > T$, $f_{j^*}(\vec{\theta} + r\vec{e}, \vec{\theta})  - f_{j^*}(\vec{\theta} + T\vec{e}, \vec{\theta})  < 0$ and $D(\theta^s, \vec{\theta}) < 0$,
          so Equation \eqref{eq:long-2} is positive.
          Consequently, we have:
          \begin{equation*}
            \mathbb{E}_{{\bm {X}}, \vec{\theta}^{s} \sim \mathcal{T}_A|\vec{\theta}}[U_{j^*}] - \mathbb{E}_{{\bm {X}}, \vec{\theta}^{s} \sim \mathcal{T}_B|\vec{\theta}}[U_{j^*}] > 0.
          \end{equation*}
          Thus, scoring against a more precise estimate (in a radial MLRP sense) strictly increases the probability of identifying the best expert $j^*$.
        \end{proof}

        \section{Experiments}

        We run all experiments on a Macbook Pro M4 with 64GB RAM. Running WOMAC once for a fixed $k$ on the GPU takes 0.2s on the ACX dataset, and takes 0.3s on the HFC dataset. We also provide additional results from our experiments:

        \begin{enumerate}
          \item \textbf{ACX data}:
            \begin{enumerate}
              \item We plot the average difference in Pearson and Spearman correlations across $d$ sub-samples for each $(m_{\text{train}}, m_{\text{test}})$ in Figure \ref{fig:oos-acx-errs}. This establishes that WOMAC scores are indeed significantly better than MSE scores at predicting experts out-of-sample performance given the same amount of data. Intuitively, it makes sense that the gap shrinks; with enough tasks, the MSE score should approach the same predictiveness as WOMAC scores. However, attention is scarce and every additional question imposes a cost on the participant. Thus, competition designers would prefer to efficiently identify the top experts with as few tasks as possible; the fact that the gap is substantial even with between 5-40 tasks shows that WOMAC is significantly more efficient within a practical range of questions.
              \item We plot the Pearson and Spearman correlations for various choices of the hyperparameter $k$ in Figure \ref{fig:oos-acx-cutoffs}. In other words, we vary $k$ and plot the Pearson and Spearman correlations when WOMAC scores are computed as the MSE against the top-$k$\% experts (scored against in-sample outcomes). We find that the choice of cutoff does not make a substantial difference. This is reasonable; the average of the top 5\% shouldn't be \emph{that} different from the average of the all experts since the errors in the latter should be uncorrelated and cancel out.
              \item We plot the Pearson and Spearman correlations using various sub-sampled number of experts, to study the effect of the number of competitors, in Figure \ref{fig:oos-acx-n}. Once again, we find no substantial difference. Once again, this stems from the fact that the reference solution, computed as the average of the top-$k$\% of experts, becomes relatively stable quickly.
            \end{enumerate}
          \item \textbf{HFC Data}:
            \begin{enumerate}
              \item We plot Pearson and Spearman correlations of in-sample WOMAC scores and MSE scores with out-of-sample MSE scores in Figure \ref{fig:oos-hfc-main} (analogous to Figure \ref{fig:oos-acx}). As with the ACX data, we observe that in-sample WOMAC score is more predictive of out-of-sample MSE score compared to in-sample MSE score.
              \item We also plot the average difference in Pearson and Spearman correlations across $d$ sub-samples for each $(m_{\text{train}}, m_{\text{test}})$ in Figure \ref{fig:oos-hfc-errs} (analogous to Figure \ref{fig:oos-acx-errs}). We again observe that WOMAC scores are significantly better than MSE scores at predicting out-of-sample MSE scores, though the gain is smaller.
              \item We plot the median, IQR, and 5th-95th percentile range of optimal cutoff hyperparameters across simulations for each training set size in Figure \ref{fig:oos-hfc-cutoffs}. We observe that roughly taking the average of the top-5\% minimizes the MSE with respect to realized outcomes.
            \end{enumerate}
        \end{enumerate}

        \begin{figure*}[h]
          \centering
          \includegraphics[scale=0.65]{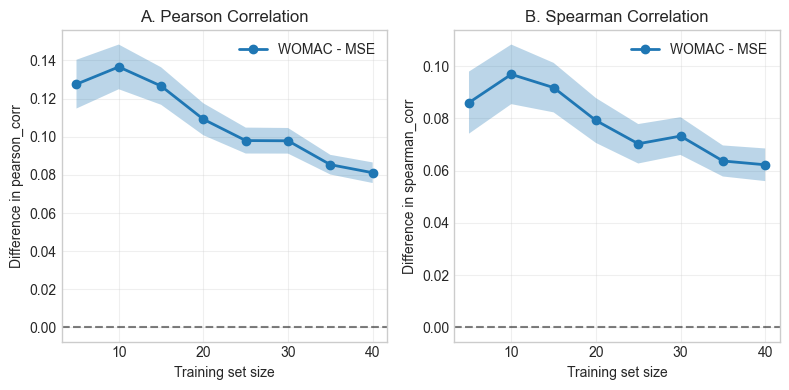}
          \caption{\textbf{[ACX Data] Plot of gap between in-sample WOMAC score's \emph{correlation} with out-of-sample MSE score, and in-sample MSE score's \emph{correlation} with out-of-sample MSE score.} The figure shows that in-sample WOMAC scores are significantly more predictive of out-of-sample MSE scores, compared to in-sample MSE scores.}
          \label{fig:oos-acx-errs}
        \end{figure*}

        \begin{figure*}[h]
          \centering
          \includegraphics[scale=0.65]{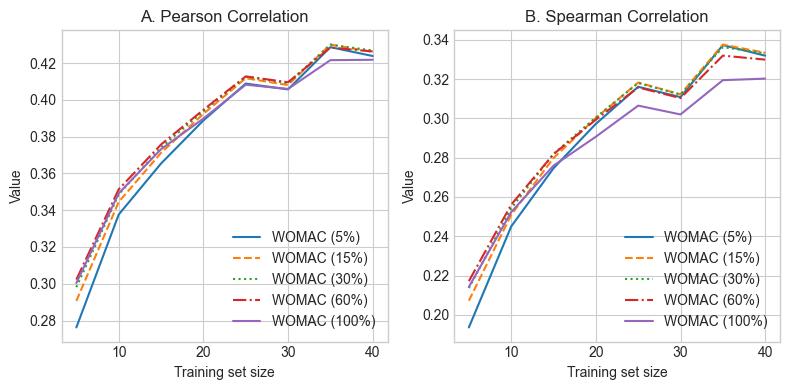}
          \caption{\textbf{[ACX Data] Plot of correlations (y-axis) between score on in-sample tasks and MSE on out-of-sample tasks, varying number of in-sample tasks (x-axis).} Different lines indicate different cutoff values of hyperparameter $k$. For example, $k=30\%$ means experts are scored against the average of the top 30\% of peers (as evaluated using MSE against in-sample outcomes). The figure shows that the exact choice of $k$ is not too material.}
          \label{fig:oos-acx-cutoffs}
        \end{figure*}

        \begin{figure*}
          \centering
          \includegraphics[scale=0.65]{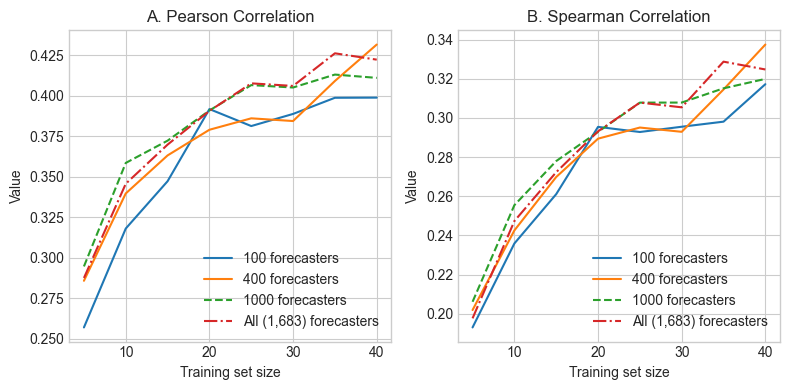}
          \caption{\textbf{[ACX Data] Plot of correlations (y-axis) between score on in-sample tasks and MSE on out-of-sample tasks, varying number of in-sample tasks (x-axis).} Different lines indicate different competitions pool sizes $n$. The figure shows predictiveness of WOMAC scores does not depend too strongly on $n$ once $n$ is reasonable large (e.g., $n=100$).}
          \label{fig:oos-acx-n}
        \end{figure*}

        \begin{figure*}
          \centering
          \includegraphics[scale=0.65]{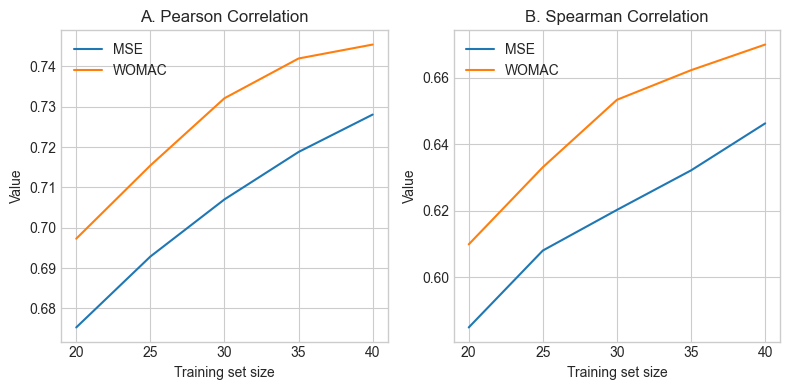}
          \caption{\textbf{[HFC Data] Plot of correlations (y-axis) between score on in-sample tasks and MSE on out-of-sample tasks, varying number of in-sample tasks (x-axis).} Analogous to Figure \ref{fig:oos-acx}, we observe that in-sample WOMAC score is more predictive of out-of-sample MSE score compared to in-sample MSE score.}
          \label{fig:oos-hfc-main}
        \end{figure*}

        \begin{figure*}
          \centering
          \includegraphics[scale=0.65]{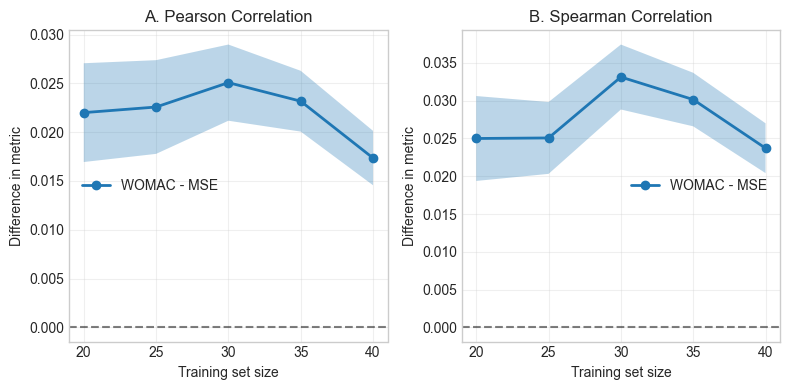}
          \caption{\textbf{[HFC Data] Plot of gap between in-sample WOMAC score's \emph{correlation} with out-of-sample MSE score, and in-sample MSE score's \emph{correlation} with out-of-sample MSE score.} Analogous to Figure \ref{fig:oos-acx-errs}, we observe that in-sample WOMAC scores are significantly more predictive of out-of-sample MSE scores, compared to in-sample MSE scores.}
          \label{fig:oos-hfc-errs}
        \end{figure*}

        \begin{figure*}
          \centering
          \includegraphics[scale=0.65]{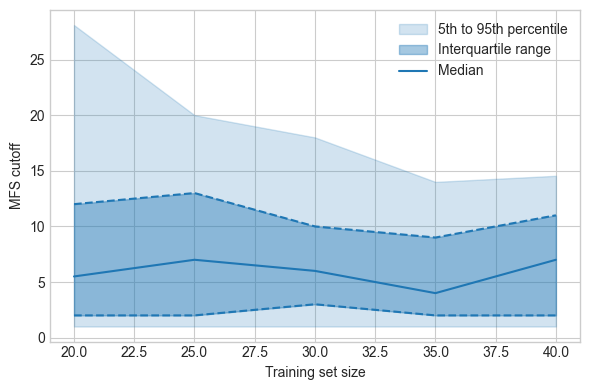}
          \caption{\textbf{[HFC Data] Plot of optimal cutoff hyperparameter $k$ for various training set sizes.} This figure shows the median, IQR, and 5th-95th percentile range of optimal cutoff hyperparameters across simulations for each training set size. We observe that roughly taking the average of the top-5\% minimizes the MSE with respect to realized outcomes.}
          \label{fig:oos-hfc-cutoffs}
        \end{figure*}

        \end{document}